\documentclass[letterpaper,UKenglish,cleveref, autoref, thm-restate,dvipsnames]{lipics-v2021}
\usepackage{amssymb, amsthm, amsmath,stmaryrd}
\usepackage[numbers]{natbib}
\usepackage{caption, subcaption}
\usepackage{graphicx,wrapfig}
\usepackage{array, tabularx, multirow, makecell}
\usepackage{tikz}
\usepackage{environ,listings}
\usepackage{color}
\usepackage{pifont}
\usepackage{textcomp}
\usepackage{hyperref}
\usepackage[shortlabels]{enumitem}
\usepackage{wrapfig}
\usepackage{booktabs}
\usepackage{siunitx}
\usepackage{pgfplotstable}
\usepackage{csvsimple}
\usepackage{booktabs}
\usepackage{colortbl}
\usepackage{caption}
\usepackage[linesnumbered,ruled,lined,boxed,noend]{algorithm2e}
\usepackage{bold-extra}
\SetNlSty{bfseries}{\color{black}}{}
\usepackage{csvsimple}
\usepackage{graphicx}
\usepackage[capitalise,noabbrev]{cleveref}
\usepackage{float}
\usepackage{xspace}
\usepackage{tikz}
\usepackage{dsfont}
\usepackage{xcolor}

\allowdisplaybreaks

\hypersetup{
    colorlinks=true,
    citecolor=blue,
    pdftitle={ApproxEM}
    }

\newtheorem*{definition*}{Definition}

\newtheorem*{invariant*}{Invariant}

\Crefname{claim}{Claim}{Claims}

\newcommand{\anita}[1]{\textcolor{Orange}{Anita: #1} }

\newcommand{\Ep}{E^+}
\newcommand{\En}{E^-}
\newcommand{\Cp}{C^+}
\newcommand{\EMO}{\textsc{EM-opt}}
\renewcommand{\subset}{\subseteq}
\newcommand{\tup}{(G, k, M^*, M)}

\newcommand{\mc}[1]{\mathcal #1}

\DeclareMathOperator*{\argmin}{arg\,min}

\SetKwSty{textbf}
\SetKwFor{For}{\normalfont{For}}{\normalfont{do the following:}}{}%

\makeatletter
\newcommand{\problemtitle}[1]{\gdef\@problemtitle{#1}}
\newcommand{\probleminput}[1]{\gdef\@probleminput{#1}}
\newcommand{\problemtask}[1]{\gdef\@problemtask{#1}}
\newcommand{\problemoutput}[1]{\gdef\@problemoutput{#1}}
\NewEnviron{problem}{
	\vspace{0.1cm}
	\problemtitle{}\probleminput{}\problemtask{}
	\BODY
	\par\addvspace{.5\baselineskip}
	\noindent
	\begin{tabularx}{\textwidth}{@{\hspace{\parindent}} l X c}
		\multicolumn{2}{@{\hspace{\parindent}}l}{\@problemtitle} \\
		\textbf{Input:} & \@probleminput \\
		\textbf{Task:} & \@problemtask
	\end{tabularx}
	\par\addvspace{.5\baselineskip}
	\vspace{0.1cm}
}

\NewEnviron{search_problem}{
	\vspace{0.1cm}
	\problemtitle{}\probleminput{}\problemoutput{}
	\BODY
	\par\addvspace{.5\baselineskip}
	\noindent
	\begin{tabularx}{\textwidth}{@{\hspace{\parindent}} l X c}
		\multicolumn{2}{@{\hspace{\parindent}}l}{\@problemtitle} \\
		\textbf{Input:} & \@probleminput \\
		\textbf{Output:} & \@problemoutput
	\end{tabularx}
	\par\addvspace{.5\baselineskip}
	\vspace{0.1cm}
}

\crefname{cond}{Condition}{Conditions}
\creflabelformat{prop}{#2#1#3\@ifnextchar.{}{.\@}}

\crefname{obs}{Observation}{Observations}
\creflabelformat{prop}{#2#1#3\@ifnextchar.{}{.\@}}

\crefname{pty}{Property}{Properties}
\creflabelformat{prop}{#2#1#3\@ifnextchar.{}{.\@}}

\crefname{inv}{Invariant}{Invariants}
\creflabelformat{prop}{#2#1#3\@ifnextchar.{}{.\@}}

\DeclareMathAlphabet\mathbfcal{OMS}{cmsy}{b}{n}

\renewcommand{\theenumi}{\arabic{enumi}}
\renewcommand{\theenumii}{\arabic{enumii}}

\renewcommand\p@enumii{\theenumi.}
\renewcommand\p@enumiii{\theenumi.\theenumi.}

\makeatother

\usetikzlibrary{decorations.markings}

\tikzset{
every picture/.style={
    line width=1.2pt, x=0.75pt,y=0.75pt,yscale=-1,xscale=1
    },
    >=stealth,
    orient/.style n args={3}{
        decoration={
            markings,
            mark=at position #3 with {
                        \arrow[line width=2pt]{#1};
            }
        }, postaction={decorate}},
    orient/.default={<}{1}{0.5}
}

\definecolor{mred}{RGB}{208, 2, 27 }
\definecolor{mgreen}{RGB}{126,211,33}

\title{An Approximation Algorithm for the Exact Matching Problem in Bipartite Graphs}
\titlerunning{Approximation of EM in Bipartite Graphs} 

\author{Anita {D\"{u}rr}}{Department of Computer Science, ETH Z\"{u}rich, Switzerland }{anita.durr@epfl.ch}{0000-0003-0440-5008}{}
\author{Nicolas {El Maalouly}}{Department of Computer Science, ETH Z\"{u}rich, Switzerland }{nicolas.elmaalouly@inf.ethz.ch}{0000-0002-1037-0203}{}
\author{Lasse {Wulf}}{Institute of Discrete Mathematics, TU Graz, Austria }{wulf@math.tugraz.at}{0000-0001-7139-4092}{Supported by the Austrian Science Fund (FWF): W1230.}

\authorrunning{A. D\"{u}rr, N. El Maalouly, L. Wulf} 

\Copyright{Anita D\"{u}rr, Nicolas El Maalouly, Lasse Wulf} 




\nolinenumbers 

\EventEditors{John Q. Open and Joan R. Access}
\EventNoEds{2}
\EventLongTitle{42nd Conference on Very Important Topics (CVIT 2016)}
\EventShortTitle{CVIT 2016}
\EventAcronym{CVIT}
\EventYear{2016}
\EventDate{December 24--27, 2016}
\EventLocation{Little Whinging, United Kingdom}
\EventLogo{}
\SeriesVolume{42}
\ArticleNo{23}
\pgfplotsset{compat=1.16}

\begin{document}
\maketitle

\begin{abstract}
    In 1982 Papadimitriou and Yannakakis introduced the \textsc{Exact Matching} problem, in which given a red and blue edge-colored graph $G$ and an integer $k$ one has to decide whether there exists a perfect matching in $G$ with exactly $k$ red edges. Even though a randomized polynomial-time algorithm for this problem was quickly found a few years later, it is still unknown today whether a deterministic polynomial-time algorithm exists. This makes the \textsc{Exact Matching} problem an important candidate to test the \textbf{RP}=\textbf{P} hypothesis.

    In this paper we focus on approximating \textsc{Exact Matching}. While there exists a simple algorithm that computes in deterministic polynomial-time an \emph{almost} perfect matching with exactly $k$ red edges, not a lot of work focuses on computing perfect matchings with \textit{almost} $k$ red edges. In fact such an algorithm for bipartite graphs running in deterministic polynomial-time was published only recently (STACS'23). It outputs a perfect matching with $k'$ red edges with the guarantee that $0.5k \leq k' \leq 1.5k$. In the present paper we aim at approximating the number of red edges without exceeding the limit of $k$ red edges. We construct a deterministic polynomial-time algorithm, which on bipartite graphs computes a perfect matching with $k'$ red edges such that $\frac{1}{3}k \leq k' \leq k$.
\end{abstract}

\section{Introduction}
In the \textsc{Exact Matching} problem (denoted as EM) one is given a graph $G$ whose edges are colored in red or blue and an integer $k$ and has to decide whether there exists a perfect matching $M$ in $G$ containing exactly $k$ red edges. This problem was first introduced in 1982 by Papadimitriou and Yannakakis \cite{Papadimitriou_Yannakakis} who conjectured it to be \textbf{NP}-complete. However a few years later, Mulmuley et al. \cite{matrix_inversion_vazirani} showed that the problem can be solved in randomized polynomial-time, making it a member of the class \textbf{RP}. This makes it unlikely to be \textbf{NP}-hard. In fact, while it is commonly believed that \textbf{RP}=\textbf{P}, the proof of this statement is a major open problem in complexity theory. Since EM is contained in \textbf{RP} but a deterministic polynomial-time algorithm for it is not known, the problem is a good candidate for testing the \textbf{\textbf{RP}}=\textbf{P} hypothesis.
Actually Mulmuley et al. \cite{matrix_inversion_vazirani} even showed that EM is contained in \textbf{RNC}, which is the class of decision problems that can be solved by a polylogarithmic-time algorithm using polynomially many parallel processors while having access to randomness. As studied in \cite{Svensson_parallel_computation_matching}, derandomizing the perfect matching problem from this complexity class is also a big open problem. This makes EM even more intriguing as randomness allows it to be efficiently parallelizable but we don't even know how to solve it sequentially without randomness.

Due to this, numerous works cite EM as an open problem. This includes the above mentioned seminal work on the parallel computation complexity of the matching problem \cite{Svensson_parallel_computation_matching}, planarizing gadgets for perfect matchings \cite{Gurjar_planarizing_gadgets_for_PM}, multicriteria optimization \cite{multicriteria_optimization}, matroid intersection \cite{matroid_intersection}, DNA sequencing \cite{DNA_sequencing}, binary linear equation systems with small hamming weight \cite{linear_equ_hamming}, recoverable robust assignment \cite{robust_assignment} as well as more general constrained matching problems \cite{budgetMatchingPTAS,Mastrolilli_Stamoulis_12,Stamoulis,Mastrolilli_Stamoulis_14}. Progress in finding deterministic algorithms for EM has only been made for very restricted classes of graphs, thus illustrating the difficulty of the problem. There exists a deterministic polynomial-time algorithm for EM in planar graphs and more generally in $K_{3, 3}$-minor free graphs \cite{Yuster_Almost_EM}, as well as for graphs of bounded genus \cite{Galluccio_Loebl}. Additionally, standard dynamic programming techniques on the tree decomposition of the graph can be used to solve EM on graphs of bounded treewidth \cite{Maalouly23_stacs, vardi2022quantum}. Contrary to those classes of sparse graphs, EM on dense graphs seems to be harder to solve. At least 4 articles study solely complete and complete bipartite graphs \cite{karzanov1987maximum,geerdes2011unified,Yi_Murty_Spera,Gurjar_Korwar_Messner_Thierauf}. More recently those results were generalized in \cite{Maalouly_Steiner_22} with a polynomial-time algorithm solving EM for constant independence number\footnote{The independence number of a graph $G$ is defined as the largest number $\alpha$ such that $G$ contains an independent set of size $\alpha$.} $\alpha$ and constant bipartite independence number\footnote{The bipartite independence number of a bipartite graph $G$ is defined as the largest number $\beta$ such that $G$ contains a balanced independent set of size $2\beta$, i.e.\  an independent set using exactly $\beta$ vertices from each partition.} $\beta$. These algorithm have \textbf{XP}-time, i.e.\ they run in time $O(g(\alpha) n^{f(\alpha)})$ and $O(g(\beta) n^{f(\beta)})$. This was further improved for bipartite graphs in \cite{Maalouly_Steiner_Wulf} where the authors showed an FPT algorithm parameterized by the bipartite independence number $\beta$ solving EM on bipartite graphs.

\paragraph*{Approximating EM.}
Given the unknown complexity status of EM, it is natural to try approximating it. In order to approximate EM, we need to allow for non-optimal solutions. Here two directions can be taken: either we consider non-perfect matchings, or we consider perfect matchings with not exactly $k$ red edges. Yuster \cite{Yuster_Almost_EM} took the first approach and showed an algorithm that computes an \emph{almost perfect exact matching} (a matching with exactly $k$ red edges and either $\frac n 2$  or $\frac n 2 - 1$ total edges\footnote{$n$ is the number of vertices in the graph.}). This is as close as possible for this type of approximation. However, there are two things to consider with Yuster's algorithm: First, the algorithm itself is very simple. Secondly, as long as a trivial condition on $k$ is met, such an almost perfect exact matching always exists. It is surprising that this approximation is achieved by such simple methods in such generality. For the closely related \emph{budgeted matching} problem, more sophisticated methods have been used recently to achieve a PTAS \cite{budgetMatchingPTAS} and efficient PTAS \cite{budgetMatchingEPTAS}. These methods also do not guarantee to return a perfect matching.

The alternative is to find approximations which always return a perfect matching, but with possibly the wrong number of red edges. These kinds of approximations seem considerably more difficult to tackle. It is puzzling how little is known about this type of approximation, while the other type of approximation admits much stronger results. Very recently, El Maalouly \cite{Maalouly23_stacs} presented a polynomial-time algorithm which for bipartite graphs outputs a perfect matching containing $k'$ red edges with $0.5k \leq k' \leq 1.5k$. To the best of our knowledge, this is the only result in this direction, despite the original EM problem being from 1982. Moreover, note that this algorithm can return both a perfect matching with too many or too few red edges. Hence El Maalouly's algorithm is not an approximation algorithm in the classical sense, as it only guarantees a two-sided, but not a one-sided, approximation.

\paragraph*{Our contribution.}

We consider the problem of approximating EM, such that a perfect matching is always returned, and such that the approximation is one-sided. We show the first positive result of this kind. Formally, we consider the following optimization variant of EM:

\begin{problem}
\problemtitle{\textsc{Exact Matching Optimization (EM-opt)}}
\probleminput{A graph $G$ whose edges are colored red or blue; and an integer $k$.}
\problemtask{Maximize $|R(M)|$ subject to the constraint that $M$ is a perfect matching in $G$ and $|R(M)| \leq k$, where $R(M)$ is the set of red edges in $M$. }
\end{problem}

We consider bipartite graphs and prove:
\begin{theorem}\label{thm:3approx}
There exists a deterministic polynomial-time 3-approximation for \textsc{EM-opt} in bipartite graphs.
\end{theorem}

We remark that there are three kinds of input instances to \textsc{EM-opt} : First, instances which are also YES-instances of EM, i.e.\ a perfect matching with $k$ red edges exists. For those instances, our algorithm returns a perfect matching with $k'$ red edges and $\frac 1 3 k \leq k' \leq k$. Second, there are instances, where \textsc{EM-opt} is infeasible, in the sense that all perfect matchings have $k + 1$ or more red edges. For such instances, our algorithm returns ``infeasible''. Finally, there are instances which are feasible, but the optimal solution to \textsc{EM-opt} has $k^*$ red edges with $k^* < k$. For such instances, our algorithm returns a perfect matching with $k'$ red edges and $\frac 1 3 k^* \leq k' \leq k^*$.

The main idea behind our approximation algorithm is to start with an initial perfect matching, and then repeatedly try a small improvement step using dynamic programming. Whenever a small improvement step is not possible, we prove that the graph is ``rigid'' in a certain sense. We can prove that as a result of this rigidity, it is a good idea to guess a single edge $e$ and to enforce that this edge $e$ is contained in the solution matching. This paper contains concepts and lemmas, which make these ideas formal and may help for the development of future approximation algorithms. Our new ideas are based on a geometric intuition about the cycles appearing in such a ``rigid'' graph.


\section{Preliminaries}
\subsection{Definitions and notations}
All graphs considered are simple. Given a graph $G=(V, E)$, a perfect matching $M$ of $G$ is a subset of $E$ such that every vertex in $V$ is adjacent to exactly one edge in $M$. We introduce two tools related to $M$: the weight function $w_M$ and the weighted directed graph $G_M$. These tools were first introduced in \cite{Maalouly_Steiner_22,Maalouly23_stacs}. They can be used to reason about EM, as is explained further below.

\begin{definition*}[Weight function $w_M$]
Given a graph $G=(V, E)$ whose edges are colored red or blue and a perfect matching $M$ on $G$, we define the weight function $w_M$ on the edges of $G$ as follows:
$$
w_M(e) = \left \{
	\begin{array}{ll}
		0  & \mbox{if } e \mbox{ is a blue edge} \\
		-1 & \mbox{if } e \in M \mbox{ is a red edge} \\
        +1 & \mbox{if } e \notin M \mbox{ is a red edge}
	\end{array}
\right.
$$
\end{definition*}
\begin{definition*}[Oriented and edge-weighted graph $G_M$]
Given a bipartite graph $G=(A \sqcup B, E)$ whose edges are colored red or blue and a perfect matching $M$ on $G$, the graph $G_M$ is the oriented and edge-weighted graph such that:
\begin{itemize}
    \item the vertex set of $G_M$ is the vertex set $A \sqcup B$ of $G$;
    \item the edge set of $G_M$ is the edge set of $G$ such that edges in $M$ are oriented from $A$ to $B$ and edges not in $M$ are oriented from $B$ to $A$;
    \item edges in $G_M$ are weighted according to the weight function $w_M$.
\end{itemize}
\end{definition*}

For ease of notation, we identify subgraphs of $G$ and $G_M$ with their edge sets. Any reference to their vertices will be made explicit.
Since every edge in $G_M$ has an undirected unweighted version in $G$, every subgraph $H_M$ in $G_M$ can be associated to the subgraph $H$ in $G$ of corresponding undirected and unweighted edge set, and vice versa. With a slight abuse of notation, we therefore sometimes do not differentiate $H$ of $H_M$ and denote them by the same object.
If $E_1$ and $E_2$ are two edge sets, then we denote by $E_1 \Delta E_2 := (E_1
\setminus E_2) \cup (E_2 \setminus E_1)$ their symmetric difference.
Let $H$ be a subgraph of $G_M$. We use the following notations:
\begin{itemize}
    \item $R(H)$ denotes the set of red edges in $H$.
    \item $\Ep(H)$ denotes the set of positive edges in $H$ (i.e.\ the red edges in $H$ not contained in $M$).
    \item $\En(H)$ denotes the set of negative edges in $H$ (i.e.\ the red edges in $H$ contained in $M$).
    \item $w_M(H)$ is the weight of $H$ with respect to the edge weight function $w_M$, i.e.\ $w_M(H) = |\Ep(H)| - |\En(H)|$.
\end{itemize}
Unless stated otherwise, all considered cycles and paths in $G_M$ are directed.
Finally if $C$ is a directed cycle in $G_M$, and $P_1
\subset C$ and $P_2 \subset C$ are sub-paths on $C$, we denote by $C[P_1, P_2]$ (respectively $C(P_1, P_2)$) the sub-path in $C$ from $P_1$ to $P_2$ included (resp. excluded). If $w_M(C) > 0$ then we call $C$ a \emph{positive} cycle and if $w_M(C) \leq 0$ we call $C$ a \emph{non-positive} cycle.
A \emph{walk} in a directed graph is a sequence of edges $(e_1,\dots,e_t)$ such that the end vertex of $e_i$ is the start vertex of $e_{i+1}$ for $i=1,\dots,t-1$. In contrast to a path, a walk may visit vertices and edges twice. A walk is \emph{closed}, if its start and end vertex are equal.

\subsection{Observations on $G_M$}
\label{subsec:observations-G-M}
In this subsection, we explain some simple observations, which we use later to reason about our approximation algorithm. Let $G$ be a bipartite graph and $M$ be a perfect matching of $G$. A cycle in $G$ is said to be \emph{$M$-alternating} if for any two adjacent edges in the cycle, one of them is in $M$ and the other is not. It is a well-known fact that if $M$ and $M'$ are two perfect matchings, then $M \Delta M'$ is a set of vertex-disjoint cycles
that are both $M$-alternating and $M'$-alternating. We now make the following important observations on the bipartite graphs $G$ and $G_M$ and the weight function $w_M$.

\begin{observation}\label{obs:directed_M_alternating}
A cycle $C$  in $G$ is  $M$-alternating if and only if the corresponding cycle $C_M$ in $G_M$ is directed.
\end{observation}
\begin{proof}
This follows directly from the definition of $G_M$.
\end{proof}

\begin{observation}\label{obs:rMDC}
Let $C$ be a directed cycle in $G_M$. Then $M':= M \Delta C$ is a perfect matching whose number of red edges is $|R(M')| = |R(M)| + w_M(C)$.
\end{observation}
\begin{proof}
    Since $C$ is $M$-alternating, it is a well-known fact that $M'$ is again a perfect matching. The equation $|R(M')| = |R(M)| + w_M(C)$, follows directly from the definition of $w_M$. If an edge in $C$ is blue, it does not change the amount of red edges of $M$. If an edge in $C$ is red, it changes the number of red edges of $M$ by $\pm 1$, depending on whether or not it is in $M$.
\end{proof}

\begin{observation}\label{obs:intersectingC}
    Let $C_1$ and $C_2$ be two directed cycles in $G_M$ that intersect at a vertex $v$. Then $C_1$ and $C_2$ also intersect on an edge adjacent to $v$.
\end{observation}

\begin{proof}
    By \cref{obs:directed_M_alternating}, $C_1$ and $C_2$ are $M$-alternating so they both contain exactly one adjacent edge to $v$ that is also contained in $M$.
    Since $M$ is a perfect matching, there exists a single adjacent edge to $v$ in $M$. Thus $C_1$ and $C_2$ both contain this edge.
\end{proof}

\subsection{Previous ideas}

The first part of our algorithm for \cref{thm:3approx} presented in \cref{sec:proof_thm} is strongly inspired by the algorithm in \cite[Theorem 1]{Maalouly23_stacs} that computes a perfect matching with between $0.5k$ and $1.5k$ red edges. We summarize the main ideas here.

One can compute a perfect matching of minimum number of red edges in polynomial time by using a maximum weight perfect matching algorithm (e.g.\ \cite{edmonds1965maximum}) on the graph $G$ where red edges have weight $0$ and blue edges have weight $+1$.
The idea in \cite{Maalouly23_stacs} is to start with a perfect matching of minimum number of red edges $M$. If $0.5k \leq |R(M)| \leq 1.5k$ then we are already done and can output $M$. Otherwise $|R(M)| < 0.5k$ (it cannot have more than $k$ red edges by minimality). In that case the algorithm iteratively improves $M$ using positive cycles. Indeed, assume for a moment that we could find a directed cycle $C$ in $G_M$ such that both $w_M(C) >0$ and $|\Ep(C)| \leq k$, then the perfect matching $M' := M \Delta C$ is such that
\begin{align*}
    |R(M)| < |R(M')| &= |R(M)| + w_M(C) \\
    &\leq |R(M)| + |\Ep(C)| < 0.5k + k = 1.5k.
\end{align*}
Thus we can iteratively compute such cycles $C$ and replace $M$ by $M \Delta C$. We make true progress every time until $0.5k \leq |R(M)| \leq 1.5k$.

An important observation proven in \cite{Maalouly23_stacs} is that we can determine in polynomial time if such a cycle exists. We recall this result in \cref{prop:recycling}.
\begin{proposition}[Adapted from {\cite[Proposition 11]{Maalouly23_stacs}}]\label{prop:recycling}
Let $G:= (V, E)$ be an edge-weighted directed graph and $t \in \mathbb{N}$ be a parameter. There exists a deterministic polynomial-time algorithm that, given $G$, determines whether or not there exists a directed cycle $C$ in $G$ with $w(C) >0$ and $|\Ep(C)| \leq t$. If such a cycle $C$ exists then the algorithm also outputs a cycle $C'$ with the same properties as $C$ (i.e.\  $w(C') > 0$ and $|E^+(C')| \leq t$).
\end{proposition}

The idea of \cref{prop:recycling} is to flip the sign of all weights so that we are looking for a negative cycle. Those can be found by shortest path algorithms: for example the Bellman-Ford algorithm which relies on a dynamic program (DP). Each entry of this DP contains the shortest distance between two vertices. The DP can be adapted so that it contains an additional budget constraint  corresponding to the number of negative edges (i.e. positive before we flip the sign of the weight) a path is allowed to use.
We also add in the entry of the DP the last edge used in the path. This way, when looking at the entries corresponding to the shortest path from a vertex $v$ to the same vertex $v$ we are able to determine whether or not there exists a negative cycle that uses no more than $t$ negative edges (i.e. positive before we flip the sign of the weight). If it is the case then the algorithm can output such a cycle.
We redirect the reader for a detailed proof of Proposition 11 in \cite[Section 3.1]{Maalouly23_stacs}. We finally remark that the proof in \cite{Maalouly23_stacs} is completed by showing that as long as $M$ does not have between $0.5k$ and $1.5k$ red edges, a cycle $C$ with $w_M(C) > 0$ and $|\Ep(C)| \leq k$ always exists. However, to get a one-sided approximation, we need to guarantee the existence of a cycle $C$ with $w_M(C) > 0$ and $|\Ep(C)| \leq k - |R(M)|$. This might not be possible for certain graphs if $|R(M)| > 0$. This means that the above result is not enough for our purpose, as it only leads to a two-sided approximation.

\subsection{Assumptions}\label{sec:assumptions}
In order to reduce the technical details needed to present our 3-approximation algorithm of $\textsc{EM-opt}$, we show in this subsection that several simplifying assumptions can be made about the input instance.

\begin{claim}\label{claim:feasibility}
    We can assume without loss of generality that the input instance to $\textsc{EM-opt}$ is also a YES-instance of EM, i.e.\ there exists a perfect matching with exactly $k$ red edges.
\end{claim}
\begin{proof}
    Fix an instance $\langle G, k \rangle$ of \EMO.
    If the instance is feasible, let $k^* \leq k$ be the number of red edges in an optimal solution of $\textsc{EM-opt}$.
    Suppose $\mathcal{A}$ is an algorithm for $\textsc{EM-opt}$, which given an instance $\langle G,k \rangle$ returns a 3-approximation solution to {\EMO} if $k = k^*$, and returns some arbitrary perfect matching otherwise.
    We devise an algorithm $\mathcal{A}'$, which is a 3-approximation of $\textsc{EM-opt}$ for the input $\langle G, k \rangle$. First, $\mathcal{A}'$ computes in polynomial time the perfect matchings $M_\text{min}$ and $M_\text{max}$ with the minimum and maximum number of red edges. If it is not the case that $|R(M_\text{min})| \leq k \leq |R(M_\text{max})|$, then $\mathcal{A}'$ returns ``infeasible''. Afterwards, $\mathcal{A}'$ calls $\mathcal{A}$ as a subroutine for all $k' \in \{|R(M_\text{min})|,\dots, k\}$ and receives a  perfect matching $M_{k'}$ each time. Finally, $\mathcal{A}'$ returns the best among all the matchings $M_{k'}$ that are feasible for $\textsc{EM-opt}$. It is easy to see that $\mathcal{A}'$ is a correct 3-approximation, since there will be an iteration with $k' = k^*$. If $\mathcal{A}$ has running time $O(f)$, then $\mathcal{A}'$ has running time $O(nf + f_\text{Mat})$, where $f_\text{Mat}$ denotes the time to deterministically compute a maximum weight perfect matching.
\end{proof}

\begin{claim}\label{claim:existMe}
    We can assume without loss of generality that for every edge $e \in E$ there exists a perfect matching of $G$ containing $e$.
\end{claim}
\begin{proof}
If an edge is not contained in any perfect matching, it is irrelevant and it can be deleted. We can therefore pre-process the graph and check for each edge in time $O(f'_\text{Mat})$ if it is irrelevant. Here $f'_\text{Mat}$ denotes the time to deterministically check if a graph has a perfect matching, e.g.\ by running a maximum weight perfect matching algorithm.
\end{proof}

\section{A 3-approximation of Exact Matching}\label{sec:proof_thm}
In this section we are proving our main result, \cref{thm:3approx}. Let $G=(A \sqcup B, E)$ be a red-blue edge-colored bipartite graph and $k$ be an integer, such that they together form an instance of {\EMO}. We work under the assumptions of  \cref{sec:assumptions}.
In \cref{algo:3approx} we show an algorithm that given $G$ and $k$ outputs a perfect matching containing between $\frac{1}{3}k$ and $k$ red edges. Let us discuss the general ideas before formally analysing the algorithm.

\begin{algorithm}
\KwIn{A bipartite graph $G = (A \sqcup B, E)$ and an integer $k \geq 0$, such that they fulfill the assumptions from \cref{sec:assumptions}.}
\KwOut{$M$, a perfect matching such that $\frac{1}{3}k \leq |R(M)| \leq k$.}
Compute a perfect matching $M$ of minimal number of red edges. \\
If $|R(M)| \geq \frac{1}{3}k$ then output $M$. \\
Otherwise compute a cycle $C$ in $G_M$ with $w_M(C) > 0$ and $|\Ep(C)| \leq \frac{2}{3}k$, or determine that no such cycle exists. \\
If such a cycle $C$ exists, set $M \gets M \Delta C$ and go to \textit{Step 2}.\\
Otherwise there are no positive cycles with $|\Ep(C)| \leq \frac{2}{3}k$ in $G_M$. \For{\upshape every red edge $e \in R(G)$} {
Compute the perfect matching $M^e$ containing $e$ of minimal number of red edges. \\
If $\frac{1}{3}k \leq |R(M^e)| \leq k$ then output $M^e$.}
Output $\bot$
\caption{The 3-approximation algorithm for {\EMO} of \cref{thm:3approx}}
\label{algo:3approx}
\end{algorithm}

We reuse the idea from \cite{Maalouly23_stacs} to iteratively improve a perfect matching $M$ by a small amount until it has the right amount of red edges. This improvement is done by finding positive cycles in $G_M$ with no more than $\frac{2}{3}k$ positive edges (using the algorithm of \cref{prop:recycling}). As opposed to the algorithm in \cite{Maalouly23_stacs}, we do not want to ``overshoot'' the number of red edges, i.e.\ the obtained perfect matching cannot have more than $k$ red edges. We therefore constrain the cycle to have $|\Ep(C)| \leq \frac{2}{3}k$. This implies $|R(M \Delta C)| \leq k$. However the issue with that approach is that unlike in \cite{Maalouly23_stacs}, we cannot guarantee that such a cycle always exists. So what do we do if we are not able find a good cycle? The answer to this question turns out to be the key insight behind our algorithm. For every edge $e$ we define
\[ M^e \in \argmin\{|R(M)| : M \text{ is a perfect matching with } e\in M\} \]
to be a perfect matching of $G$ containing $e$ with minimal number of red edges. Note that by \cref{claim:existMe}, $M^e$ is properly defined. We also note that $M^e$ can be computed in polynomial time.

We claim that in the case that no good cycle is found, the set of matchings $M^e$ ``magically'' fixes our problems. Specifically, we claim that at least one of the matchings $M^e$ is a sufficient approximation. The rough intuition behind this is the following. The fact that no good cycle exists means that the graph does not contain small substructures, which can be used to modify the current solution $M$ towards a better approximation. In a sense, the graph is rigid. This rigidity means that maybe there could exist some edge $e$ in the optimal solution $M^*$, such that every perfect matching including $e$ is already quite similar to $M^*$. Our approximation algorithm simply tries to guess this edge $e$. After proving \cref{lem:correctness}, the remaining part of the paper is devoted to quantify and prove this structural statement.

\begin{lemma}\label{lem:correctness}
    \cref{algo:3approx} runs in polynomial time and either outputs $\bot$ or a correct 3-approximation of \EMO.
\end{lemma}
\begin{proof}
    Note that the first inner loop is executed at most $O(|V|)$ times, since $w_M(C) > 0$ in each iteration, thus $|R(M)|$ is monotonously increasing. Furthermore, every iteration of the loop is executed in polynomial time due to \cref{prop:recycling}. The second loop is executed at most $O(|E|)$ times. The matching in \textit{Step 1}, as well as the matchings $M^e$ can be computed in polynomial time using a minimum weight perfect matching algorithm.
    In total, this takes polynomial time. If the algorithm does not output $\bot$, then it outputs either in \textit{Step 7} (which is obviously a 3-approximation), or in \textit{Step 2}, which is a 3-approximation because $k/3 \leq |R(M \Delta C)| = |R(M)| + w_M(C) \leq |R(M)| + |\Ep(C)| \leq k$.
\end{proof}

All it remains to prove is that \cref{algo:3approx} never outputs $\bot$. This happens if the conditions in \textit{Step 2} and in \textit{Step 7} are not satisfied. In that regard, define the following \emph{critical tuple}.

\begin{definition*}[Critical tuple]
    Let $\tup$ be such that $G$ is a bipartite graph whose edges are colored red or blue, $k \geq 0$ is an integer and $M^*$ and $M$ are perfect matchings of $G$.
    The tuple $\tup$ is said to be \emph{critical} if:
    \begin{itemize}
        \item All the assumptions in \cref{sec:assumptions} hold.
        \item $M^*$ is a perfect matching of $G$ with exactly $k$ red edges.
        \item $|R(M)| < \frac{1}{3}k$.
        \item If $C$ is a directed cycle in $G_M$ such that $w_M(C) > 0$ then $|\Ep(C)| > \frac{2}{3}k$.
        \item For every edge ${e \in R(M^* \setminus M)}$ we have $|R(M^e)| < \frac{1}{3}k$.
    \end{itemize}
\end{definition*}


Note that if \cref{algo:3approx} outputs $\bot$, we must have a critcal tuple. Indeed, for every edge $e \in R(M^* \setminus M)$, by minimality of $M^e$, we have $|R(M^e)| \leq k$. So if the algorithm returns $\bot$, then it must be the case that $|R(M^e)| < k/3$ for all $e \in R(G)$ so in particular also for all $e \in R(M^* \setminus M)$. We will prove the following: critical tuples do not exist. This means that \cref{algo:3approx} never outputs $\bot$ and is therefore a correct 3-approximation algorithm.

\begin{lemma}\label{lem:good_Me}
    A tuple $\tup$ can never be critical.
\end{lemma}

\subsection{Overview of the main proof}\label{sec:proof_lem_Me}
In this section, we explain the idea behind the proof of our main lemma, \cref{lem:good_Me}. From this point onward, consider a fixed tuple $\tup$. We prove \cref{lem:good_Me} by contradiction and therefore assume that $\tup$ is a critical tuple. The main strategy of the proof is to find a set of cycles with contradictory properties. We call such a set of cycles a \emph{target set}. We start with a key observation stating that there only exist two different types of cycles in a critical tuple. These two types can be distinguished by examining $|\Ep(C)|$ of a cycle $C$. We will repeatedly make use of this simple observation.

\begin{observation}\label{obs:2types_cycles}
Let $\tup$ be a critical tuple. Then the following properties hold.
    \begin{enumerate}[ref=\ref{obs:2types_cycles}.\arabic*]
        \item $|\En(G_M)| < \frac{1}{3}k$. \label[obs]{pty:negG}
        \item If $C$ is a directed cycle in $G_M$ then $w_M(C) >0$ if and only if $|\Ep(C)| > \frac{2}{3}k$. \label[obs]{pty:posC}
        \item If $C$ is a directed cycle in $G_M$ then $w_M(C) \leq 0$ if and only if $|\Ep(C)| < \frac{1}{3}k$. \label[obs]{pty:negC}
    \end{enumerate}
\end{observation}
\begin{proof}
    $\tup$ is a critical tuple so in particular $|R(M)| < \frac{1}{3}k$, thus the graph $G_M$ contains at most $\frac{1}{3}k$ negative edges. So \cref{pty:negG} holds.
    \cref{pty:posC} follows directly from the definition of a critical tuple and \cref{pty:negG}.
    Finally if $w_M(C) \leq 0$ then $|\Ep(C)| \leq |\En(C)|$. By \cref{pty:negG}, $|\En(C)| < \frac{1}{3}k$, thus $|\Ep(C)| < \frac{1}{3}k$.
\end{proof}

We introduce a special cycle, which we call the cycle $\Cp$.
\begin{definition*}[Cycle $\Cp$]
    $\Cp$ is the positive cycle in $G_M$ such that $\Cp \subset M \Delta M^*$.
\end{definition*}
We claim that $\Cp$ is well-defined and unique. Indeed, since $|R(M)| < |R(M^*)|$, the set of cycles $M \Delta M^*$ has to contain at least one positive cycle. Furthermore, there are at most $k = |R(M^*)|$ positive edges in $M \Delta M^*$ and by \cref{obs:2types_cycles}, every positive cycle contains at least $\frac 2 3 k$ of them. Hence there is no second positive cycle in $M \Delta M^*$.

\begin{observation}\label{claim:Cp_bounds}
    If $\tup$ is a critical tuple, then $\frac{2}{3}k < |\Ep(\Cp)| \leq k$.
\end{observation}
\begin{proof}
    The lower bound directly follows from \cref{pty:posC}.
    Positive edges in $M \Delta M^*$ are red edges in $M^*$, hence $M \Delta M^*$ contains at most $k = |R(M^*)|$ positive edges. Thus in particular $|\Ep(C)| \leq k$.
\end{proof}

\begin{definition*}[Target set]
    We call a set $\mathcal{C}$ of non-positive cycles in $G_M$ a \emph{target set} if:
    \begin{enumerate}
        \item $\forall e \in \Ep(\Cp)$ there exists a cycle in $\mathcal{C}$ containing $e$. (\cref{cond:1cover_ep})\label[cond]{cond:1cover_ep}
        \item $\forall C \in \mathcal{C},\, C \cap \Cp$ is a single path. (\cref{cond:4single_path})\label[cond]{cond:4single_path}
        \item $\forall e \in \En(G_M)$ there are at most two cycles in $\mathcal{C}$ containing $e$. (\cref{cond:34cover_en})\label[cond]{cond:34cover_en}
    \end{enumerate}
\end{definition*}
The following \cref{lem:contradiction} shows that in order to prove our main lemma, \cref{lem:good_Me}, it suffices to find a target set.
\begin{lemma}\label{lem:contradiction}
    Let $\tup$ be a critical tuple and $\mathcal{C}$ be a target set in $G_M$. Then a contradiction arises.
\end{lemma}
\begin{proof}
    We can lower bound the sum of the weights of cycles in $\mathcal{C}$.
    \begin{align*}
    \sum_{C \in \mathcal{C}} w_M(C) &\geq 1 \cdot |\Ep(C^+)| - 2 \cdot |\En(G_M)| \\
    &> \frac{2}{3}k - 2 \cdot \frac{1}{3}k \qquad \qquad \qquad \qquad \text{(by \cref{claim:Cp_bounds,pty:negG})}\\
    &= 0
    \end{align*}
    However $\sum_{C \in \mathcal{C}} w_M(C) \leq 0$ since all cycles $C \in \mathcal{C}$ have non-positive weight $w_M(C) \leq 0$, so this is a contradiction.
\end{proof}
Note that \cref{cond:4single_path} is not needed in the proof of \cref{lem:contradiction}. However \cref{cond:4single_path} will be useful to achieve \cref{cond:34cover_en}.

The rest of the paper is structured as follows: First, we construct a set $\mathcal{C}$ of non-positive cycles satisfying only \cref{cond:1cover_ep}. This is explained in \cref{subsec:initial-set}. After that, we explain in \cref{subsec:make-simple} how to modify the initial set to additionally satisfy \cref{cond:4single_path}. In \cref{subsec:reduce-cover} we show how to modify this set again to also satisfy \cref{cond:34cover_en}. Finally, the proof of \cref{lem:good_Me} is summarized in \cref{subsec:summarize-proof-at-the-end}. To help visualize the reasonings on $G_M$  we provide various illustrations. We focus only on the orientation of paths and cycles and not on actual vertices and edges of $G_M$ so we omit them and draw paths and cycles using simple lines whose orientation is indicated by an arrow. Individual edges and vertices will be indicated in red.

\subsection{Obtaining an initial set of cycles satisfying Condition {\ref{cond:1cover_ep}}}
\label{subsec:initial-set}
We explain how to obtain an initial set of non-positive cycles in $G_M$ which satisfies only \cref{cond:1cover_ep}. To this end, for a fixed critical tuple $\tup$ consider the following definition:

\begin{definition*}[Cycle $C_e$]
    Let $e \in \Ep(\Cp)$.
    We define $C_e \subseteq M \Delta M^e$ as the unique directed cycle in $M \Delta M^e$ which contains the edge $e$.
\end{definition*}

Observe that $C_e$ is well-defined, since $e \in M \Delta M^e$. We claim that for all possible $e \in \Ep(\Cp)$ we have $w_M(C_e) \leq 0$. Indeed,
positive edges in $M \Delta M^e$ are red edges in $M^e$. By the definition of a critical tuple and $\Ep(\Cp) \subseteq R(M^* \setminus M)$,
-we have $|R(M^e)| < \frac{1}{3}k$, thus $|\Ep(C_e)|<\frac{1}{3}k$. Hence by \cref{pty:posC}, $C_e$ has non-positive weight in $G_M$.
As a direct consequence, the set of cycles $\{ C_e : e \in \Ep(\Cp) \}$ is a set of non-positive cycles satisfying \cref{cond:1cover_ep}.

\subsection{Modifying the set to satisfy Condition \ref{cond:4single_path}}
\label{subsec:make-simple}
We now show how to modify the initial set $\{C_e : e \in \Ep(\Cp)\}$ into a set of non-positive cycles such that both \cref{cond:1cover_ep,cond:4single_path} hold. We remark that this subsection is the most technical part of our paper. Except for the following \cref{lem:keep_atmost2}, the arguments in this subsection are independent from the rest of the paper, so the reader may choose to skip it and assume its main result, \cref{lem:get14}, is true. 
\cref{lem:keep_atmost2} shows with a simple argument that if a set of paths in $\Cp$ covers the whole cycle $\Cp$, then it is enough to keep at most two paths covering the same edge in $\Cp$. 
\begin{observation}\label{lem:keep_atmost2}
If $\mathcal{P'}$ is a set of paths in $\Cp$, then there exists a subset $\mathcal{P} \subseteq \mathcal{P'}$ such that $\bigcup_{P' \in \mathcal{P'}} P' = \bigcup_{P \in \mathcal{P}} P$ and every $e \in \Cp$ is contained in at most two paths of  $\mathcal{P}$.
\end{observation}
\begin{proof}
Let $X = \bigcup_{P' \in \mathcal{P'}} P'$. Take $\mathcal{P}$ to be a minimal subset of $\mathcal{P'}$ such that $\bigcup_{P \in \mathcal{P}} P = X$. Suppose there exists an edge $e \in C^+$ such that $\exists P_1, P_2, P_3 \in \mathcal{P}$ with $e \in P_1 \cap P_2 \cap P_3$. W.l.o.g.\ suppose $P_1$ contains an edge that is furthest before $e$ on $C^+$ among all edges of $P_1$, $P_2$ and $P_3$, and $P_2$ contains an edge furthest after $e$. Then $P_3 \subset P_1 \cup P_2$, which contradicts the minimality of $\mathcal{P}$. So every edge is contained in at most two paths from $\mathcal{P}$.
\end{proof}

For the rest of this subsection, we consider a fixed critical tuple $\tup$. To obtain our goal, \cref{cond:4single_path}, we would need every cycle $C_e$ to intersect $\Cp$ in a single path. However, in reality the interaction between $C_e$ and $\Cp$ can be a great amount more complicated. One example is depicted in \cref{fig:ex_Ce}.
In order to analyse the interaction between these two cycles, we define multiple notions around them. First,
we decompose the cycle $C_e$ into \emph{jumps} and \emph{interjumps} where jumps are the sub-paths of $C_e$ outside of $\Cp$ and interjumps are the sub-paths of $C_e$ intersecting $\Cp$ between jumps.
\begin{figure}[h]
    \centering
\begin{tikzpicture}

\draw  [orient={<}{1}{0.55}] (60.09,510.29) .. controls (60.09,471.66) and (91.41,440.34) .. (130.04,440.34) .. controls (168.68,440.34) and (200,471.66) .. (200,510.29) .. controls (200,548.93) and (168.68,580.25) .. (130.04,580.25) .. controls (91.41,580.25) and (60.09,548.93) .. (60.09,510.29) -- cycle ;

\begin{scope}[line width=1.5pt]
\draw  [draw opacity=0] (88.64,566.67) .. controls (82.88,568.98) and (76.59,570.25) .. (70,570.25) .. controls (42.39,570.25) and (20,547.9) .. (20,520.33) .. controls (20,492.76) and (42.39,470.41) .. (70,470.41) .. controls (70.73,470.41) and (71.46,470.43) .. (72.19,470.46) -- (70,520.33) -- cycle ; \draw  [orient={<}{1}{0.3},   color={rgb, 255:red, 74; green, 144; blue, 226 }  ,draw opacity=1 ] (88.64,566.67) .. controls (82.88,568.98) and (76.59,570.25) .. (70,570.25) .. controls (42.39,570.25) and (20,547.9) .. (20,520.33) .. controls (20,492.76) and (42.39,470.41) .. (70,470.41) .. controls (70.73,470.41) and (71.46,470.43) .. (72.19,470.46) ;  
\draw  [draw opacity=0] (62.2,494.05) .. controls (48.91,486.23) and (40,471.78) .. (40,455.25) .. controls (40,430.4) and (60.15,410.25) .. (85,410.25) .. controls (104.76,410.25) and (121.55,422.99) .. (127.6,440.71) -- (85,455.25) -- cycle ; \draw  [orient={<}{1}{0.8},   color={rgb, 255:red, 74; green, 144; blue, 226 }  ,draw opacity=1 ] (62.2,494.05) .. controls (48.91,486.23) and (40,471.78) .. (40,455.25) .. controls (40,430.4) and (60.15,410.25) .. (85,410.25) .. controls (104.76,410.25) and (121.55,422.99) .. (127.6,440.71) ;  
\draw  [draw opacity=0] (187.49,470.87) .. controls (189.94,470.46) and (192.44,470.25) .. (195,470.25) .. controls (219.85,470.25) and (240,490.4) .. (240,515.25) .. controls (240,540.1) and (219.85,560.25) .. (195,560.25) .. controls (190.03,560.25) and (185.24,559.44) .. (180.77,557.95) -- (195,515.25) -- cycle ; \draw [orient, color={rgb, 255:red, 74; green, 144; blue, 226 }  ,draw opacity=1 ] (187.49,470.87) .. controls (189.94,470.46) and (192.44,470.25) .. (195,470.25) .. controls (219.85,470.25) and (240,490.4) .. (240,515.25) .. controls (240,540.1) and (219.85,560.25) .. (195,560.25) .. controls (190.03,560.25) and (185.24,559.44) .. (180.77,557.95) ;  
\draw  [draw opacity=0] (181.16,558.12) .. controls (168.38,571.77) and (150.21,580.29) .. (130.04,580.29) .. controls (114.26,580.29) and (99.7,575.07) .. (87.99,566.26) -- (130.04,510.29) -- cycle ; \draw  [orient={<}{1}{0.65},   color={rgb, 255:red, 74; green, 144; blue, 226 }  ,draw opacity=1 ] (181.16,558.12) .. controls (168.38,571.77) and (150.21,580.29) .. (130.04,580.29) .. controls (114.26,580.29) and (99.7,575.07) .. (87.99,566.26) ;  
\draw  [draw opacity=0] (127.24,440.34) .. controls (128.09,440.31) and (128.94,440.29) .. (129.79,440.29) .. controls (154.04,440.29) and (175.41,452.63) .. (187.98,471.36) -- (129.79,510.29) -- cycle ; \draw  [orient,   color={rgb, 255:red, 74; green, 144; blue, 226 }  ,draw opacity=1 ] (127.24,440.34) .. controls (128.09,440.31) and (128.94,440.29) .. (129.79,440.29) .. controls (154.04,440.29) and (175.41,452.63) .. (187.98,471.36) ;   
\draw  [draw opacity=0] (65.7,537.91) .. controls (62.06,529.44) and (60.04,520.1) .. (60.04,510.29) .. controls (60.04,504.51) and (60.74,498.9) .. (62.07,493.52) -- (130.04,510.29) -- cycle ; \draw  [orient, color={rgb, 255:red, 74; green, 144; blue, 226 }  ,draw opacity=1] (65.7,537.91) .. controls (62.06,529.44) and (60.04,520.1) .. (60.04,510.29) .. controls (60.04,504.51) and (60.74,498.9) .. (62.07,493.52) ;  
\draw  [draw opacity=0] (102.91,445.55) .. controls (119,453.86) and (130,470.64) .. (130,490) .. controls (130,517.61) and (107.61,540) .. (80,540) .. controls (74.9,540) and (69.97,539.24) .. (65.34,537.82) -- (80,490) -- cycle ; \draw  [orient,   color={rgb, 255:red, 74; green, 144; blue, 226 }  ,draw opacity=1 ] (102.91,445.55) .. controls (119,453.86) and (130,470.64) .. (130,490) .. controls (130,517.61) and (107.61,540) .. (80,540) .. controls (74.9,540) and (69.97,539.24) .. (65.34,537.82) ; 
 
\draw  [draw opacity=0] (72.02,471.1) .. controls (79.68,459.8) and (90.56,450.87) .. (103.33,445.6) -- (130,510.34) -- cycle ; \draw  [orient,   color={rgb, 255:red, 74; green, 144; blue, 226 }  ,draw opacity=1 ] (72.02,471.1) .. controls (79.68,459.8) and (90.56,450.87) .. (103.33,445.6) ;

\end{scope}

\draw  [draw opacity=0] (165.97,570.4) .. controls (159.76,574.13) and (152.92,576.92) .. (145.64,578.58) -- (130,510.34) -- cycle ; \draw  [orient={<}{1}{0.7}, color={rgb, 255:red, 208; green, 2; blue, 27 }  ,draw opacity=1 ] (165.97,570.4) .. controls (159.76,574.13) and (152.92,576.92) .. (145.64,578.58) ;  

\draw (38,508) node [anchor=north west][inner sep=0.75pt]    {$P_{3}$};
\draw (95,489) node [anchor=north west][inner sep=0.75pt]    {$Q_{3}$};

\draw (168.5,515.5) node [anchor=north west][inner sep=0.75pt]    {$C^{+}$};
\draw (158.25,575.84) node [anchor=north west][inner sep=0.75pt]    {$e$};
\draw (227.75,449.25) node [anchor=north west][inner sep=0.75pt]    {$C_{e}$};
\draw (162.5,425) node [anchor=north west][inner sep=0.75pt]    {$P_{2}$};
\draw (245,509.25) node [anchor=north west][inner sep=0.75pt]    {$Q_{1}$};
\draw (100.5,389) node [anchor=north west][inner sep=0.75pt]    {$Q_{2}$};
\draw (59,435) node [anchor=north west][inner sep=0.75pt]    {$P_{4}$};
\draw (25,568.25) node [anchor=north west][inner sep=0.75pt]    {$Q_{4}$};
\draw (119,588.25) node [anchor=north west][inner sep=0.75pt]    {$P_{1}$};

\end{tikzpicture}
\caption{Decomposition of the cycle $C_e$ (in blue) into 4 jumps $Q_1, Q_2, Q_3, Q_4$ and 4 interjumps $P_1, P_2, P_3, P_4$. The first interjump $P_1$ contains the edge $e$}\label{fig:ex_Ce}
\end{figure}
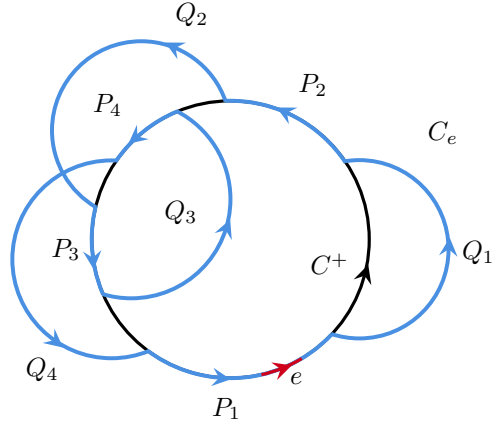
\begin{definition*}[Jumps and interjumps]
Let $e \in \Ep(\Cp)$.
A jump of $C_e$ is a sub-path $Q \subset C_e$ such that its endpoints are vertices of $\Cp$ but none of its inner vertices are in $\Cp$. Note that $Q \subset G_M \setminus \Cp$.  An interjump of $C_e$ is a sub-path of $P \subset C_e$ contained in $\Cp$ such that its endpoints are the endpoints of jumps of $C_e$. Note that $P \subset \Cp$. 
\end{definition*}
\begin{definition*}[Decomposition of $C_e$ into alternating jumps and interjumps]
    Let $e \in \Ep(\Cp)$.
    Decompose $C_e$ into jumps $Q_1, \dots, Q_\ell$ and interjumps $P_1, \dots P_\ell$ such that $e \in P_1$ and $C_e$ is the concatenation of alternating jumps and interjumps:
    $$
    C_e = P_1 Q_1 \dots P_j Q_j \dots  P_\ell Q_\ell
    $$
    Let $\mathcal{J}_e = \{Q_1, Q_2, \dots, Q_\ell\}$ be the set of jumps in $C_e$ and $\mathcal{I}_e = \{P_1, P_2, \dots, P_\ell\}$ be the set of interjumps in $C_e$.
\end{definition*}
Note that every jump in $C_e$ must be followed by an interjump since all cycles considered are $M$-alternating cycles which cannot intersect in a single vertex (by \cref{obs:intersectingC}). 
\cref{fig:ex_Ce} shows a cycle $C_e$ and its decomposition into jumps and interjumps.

To understand the interplay between $C_e$ and $\Cp$, we introduce the notion of forward and backward motion based on a visual intuition. Consider the directed cycle $\Cp$ and take the convention of always drawing it with an anticlockwise orientation.
We say that moving along $\Cp$ in the direction of its directed edges equals to an \emph{anticlockwise} or \emph{forward} motion, while moving on the cycle $\Cp$ against the direction of its directed edges equals to a \emph{clockwise} or \emph{backward} motion. 
We also want to interpret the direction a jump $Q \in \mc J_e$ is taking. However since $Q$ is not a path in $\Cp$ there is no obvious rule whether it should be seen as moving forward or backward. In fact, \cref{fig:forw_back_intuition} shows that the same jump can in principle be interpreted either as following a forward or a backward motion. We introduce the following rule which classifies all jumps as either a forward or a backward jump. We follow the convention to draw forward jumps outside of $\Cp$ and as anticlockwise arrows, while we draw backward jumps inside of $\Cp$ and as clockwise arrows. With this interpretation, in \cref{fig:ex_Ce} the jumps $Q_1, Q_2$ and $Q_4$ are forward and the jump $Q_3$ is backward. 
\input{figures/forw_back_intuition}

\begin{definition*}[Cycle $C_Q$ and forward/backward jump]
    Let $e \in \Ep(\Cp)$ and let $Q$ be a jump of $C_e$. We define $C_Q$ to be the unique directed cycle in the graph $\Cp \cup Q$ containing $Q$. Then:
    \begin{itemize}
        \item if $w_M(C_Q) > 0$, we call $Q$ a \emph{forward} jump;
        \item if $w_M(C_Q) \leq 0$, we call $Q$ a \emph{backward} jump. 
    \end{itemize}
    Let $\mathcal{J}_e^f$ be the set of forward jumps in $C_e$ and $\mathcal{J}_e^b$ the set of backward jumps in $C_e$.
\end{definition*}
Observe that this is a valid definition, since the cycle $C_Q$ is unique and independent of whether $Q$ is a forward or a backward jump. 
We give a short informal explanation of why we use the weight of $C_Q$ to distinguish between forward and backward jumps.
A positive cycle must contain a large number of positive edges (at least $\frac{2}{3}k$ many, by \cref{pty:posC}) and most of the positive edges of $C_Q$ are contained in $C^+$ (since $C_e$ has at most $\frac{1}{3}k$ positive edges).
So for a forward jump $Q_f$, the cycle $C_{Q_f}$ must have a large intersection with $C^+$ to be positive, i.e.\ $Q_f$ only skips a small portion of the cycle. Similarly, for a backward jump $Q_b$, the cycle $C_{Q_b}$ must be a non-positive cycle so it cannot contain too many positive edges and therefore cannot intersect a large part of $C^+$.
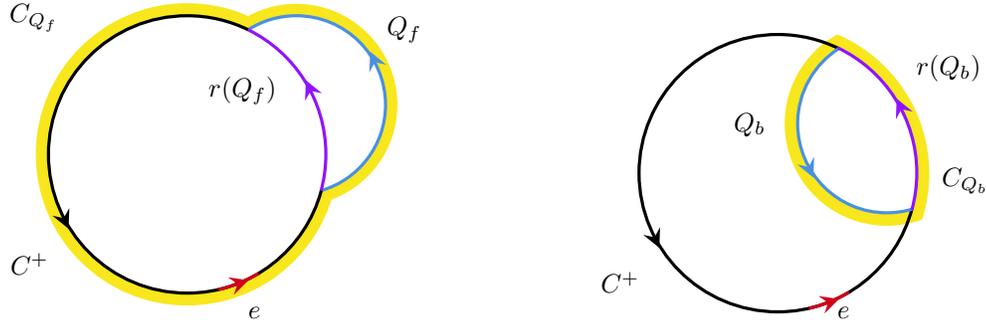
\begin{figure}[h]
    \centering
\begin{subfigure}[b]{0.45\textwidth}
\centering
\begin{tikzpicture}

\draw  [draw opacity=0][line width=4.5]  (139.8,314.19) .. controls (147.12,309.66) and (155.76,307.04) .. (165,307.04) .. controls (191.5,307.04) and (212.99,328.53) .. (212.99,355.04) .. controls (212.99,376.97) and (198.28,395.46) .. (178.18,401.19) -- (165,355.04) -- cycle ; \draw  [color={rgb, 255:red, 248; green, 231; blue, 28 }  ,draw opacity=1 ][line width=4.5]  (139.8,314.19) .. controls (147.12,309.66) and (155.76,307.04) .. (165,307.04) .. controls (191.5,307.04) and (212.99,328.53) .. (212.99,355.04) .. controls (212.99,376.97) and (198.28,395.46) .. (178.18,401.19) ;  
\draw  [draw opacity=0][line width=4.5]  (180.63,399.35) .. controls (172.14,430.44) and (143.7,453.29) .. (109.91,453.29) .. controls (69.44,453.29) and (36.62,420.48) .. (36.62,380) .. controls (36.62,339.52) and (69.44,306.71) .. (109.91,306.71) .. controls (121.87,306.71) and (133.15,309.57) .. (143.12,314.64) -- (109.91,380) -- cycle ; \draw  [color={rgb, 255:red, 248; green, 231; blue, 28 }  ,draw opacity=1 ][line width=4.5]  (180.63,399.35) .. controls (172.14,430.44) and (143.7,453.29) .. (109.91,453.29) .. controls (69.44,453.29) and (36.62,420.48) .. (36.62,380) .. controls (36.62,339.52) and (69.44,306.71) .. (109.91,306.71) .. controls (121.87,306.71) and (133.15,309.57) .. (143.12,314.64) ;  
\draw  [draw opacity=0] (141.21,316.83) .. controls (148.11,312.52) and (156.27,310.04) .. (165,310.04) .. controls (189.85,310.04) and (210,330.18) .. (210,355.04) .. controls (210,375.55) and (196.27,392.86) .. (177.51,398.28) -- (165,355.04) -- cycle ; \draw  [orient, color={rgb, 255:red, 74; green, 144; blue, 226 }  ,draw opacity=1 ] (141.21,316.83) .. controls (148.11,312.52) and (156.27,310.04) .. (165,310.04) .. controls (189.85,310.04) and (210,330.18) .. (210,355.04) .. controls (210,375.55) and (196.27,392.86) .. (177.51,398.28) ;  
\draw  [draw opacity=0] (177.63,398.17) .. controls (169.64,428.04) and (142.39,450.04) .. (110,450.04) .. controls (71.34,450.04) and (40,418.7) .. (40,380.04) .. controls (40,341.38) and (71.34,310.04) .. (110,310.04) .. controls (121.06,310.04) and (131.52,312.6) .. (140.82,317.17) -- (110,380.04) -- cycle ; \draw [orient]  (177.63,398.17) .. controls (169.64,428.04) and (142.39,450.04) .. (110,450.04) .. controls (71.34,450.04) and (40,418.7) .. (40,380.04) .. controls (40,341.38) and (71.34,310.04) .. (110,310.04) .. controls (121.06,310.04) and (131.52,312.6) .. (140.82,317.17) ;  
\draw  [draw opacity=0] (140.81,317.17) .. controls (163.97,328.58) and (179.91,352.43) .. (179.91,380) .. controls (179.91,386.32) and (179.08,392.45) .. (177.51,398.28) -- (109.91,380) -- cycle ; \draw  [orient, color={rgb, 255:red, 144; green, 19; blue, 254 }  ,draw opacity=1 ] (140.81,317.17) .. controls (163.97,328.58) and (179.91,352.43) .. (179.91,380) .. controls (179.91,386.32) and (179.08,392.45) .. (177.51,398.28) ;  
\draw  [draw opacity=0] (145.97,440.1) .. controls (139.76,443.83) and (132.92,446.62) .. (125.64,448.28) -- (110,380.04) -- cycle ; \draw  [orient={<}{1}{0.7}, color={rgb, 255:red, 208; green, 2; blue, 27 }  ,draw opacity=1 ] (145.97,440.1) .. controls (139.76,443.83) and (132.92,446.62) .. (125.64,448.28) ;  

\draw (209,309.04) node [anchor=north west][inner sep=0.75pt]    {$Q_{f}$};
\draw (119,339.04) node [anchor=north west][inner sep=0.75pt]    {$r( Q_{f})$};
\draw (19,428.04) node [anchor=north west][inner sep=0.75pt]    {$C^{+}$};
\draw (18.5,302.29) node [anchor=north west][inner sep=0.75pt]    {$C_{Q_{f}}{}$};
\draw (139,455) node [anchor=north west][inner sep=0.75pt]    {$e$};

\end{tikzpicture}
\end{subfigure}
\hfill
\begin{subfigure}[b]{0.45\textwidth}
\centering
\begin{tikzpicture}

\draw  [draw opacity=0][line width=4.5]  (521.82,530.8) .. controls (516.58,532.73) and (510.91,533.79) .. (505,533.79) .. controls (478.19,533.79) and (456.46,512.06) .. (456.46,485.25) .. controls (456.46,467.02) and (466.51,451.14) .. (481.37,442.84) -- (505,485.25) -- cycle ; \draw  [color={rgb, 255:red, 248; green, 231; blue, 28 }  ,draw opacity=1 ][line width=4.5]  (521.82,530.8) .. controls (516.58,532.73) and (510.91,533.79) .. (505,533.79) .. controls (478.19,533.79) and (456.46,512.06) .. (456.46,485.25) .. controls (456.46,467.02) and (466.51,451.14) .. (481.37,442.84) ;  
\draw  [draw opacity=0][line width=4.5]  (479.06,443.05) .. controls (505.13,454.3) and (523.38,480.25) .. (523.38,510.45) .. controls (523.38,518.08) and (522.22,525.43) .. (520.06,532.34) -- (450,510.45) -- cycle ; \draw  [color={rgb, 255:red, 248; green, 231; blue, 28 }  ,draw opacity=1 ][line width=4.5]  (479.06,443.05) .. controls (505.13,454.3) and (523.38,480.25) .. (523.38,510.45) .. controls (523.38,518.08) and (522.22,525.43) .. (520.06,532.34) ;  
\draw  [draw opacity=0] (517.75,527.91) .. controls (509.93,558.02) and (482.56,580.25) .. (450,580.25) .. controls (411.34,580.25) and (380,548.91) .. (380,510.25) .. controls (380,471.59) and (411.34,440.25) .. (450,440.25) .. controls (461.08,440.25) and (471.57,442.83) .. (480.88,447.41) -- (450,510.25) -- cycle ; \draw  [orient] (517.75,527.91) .. controls (509.93,558.02) and (482.56,580.25) .. (450,580.25) .. controls (411.34,580.25) and (380,548.91) .. (380,510.25) .. controls (380,471.59) and (411.34,440.25) .. (450,440.25) .. controls (461.08,440.25) and (471.57,442.83) .. (480.88,447.41) ;  
\draw  [draw opacity=0] (485.97,570.31) .. controls (479.76,574.04) and (472.92,576.84) .. (465.64,578.5) -- (450,510.25) -- cycle ; \draw  [orient={<}{1}{0.7}, color={rgb, 255:red, 208; green, 2; blue, 27 }  ,draw opacity=1 ] (485.97,570.31) .. controls (479.76,574.04) and (472.92,576.84) .. (465.64,578.5) ;  
\draw  [draw opacity=0] (517.93,528.36) .. controls (513.84,529.59) and (509.5,530.25) .. (505,530.25) .. controls (480.15,530.25) and (460,510.1) .. (460,485.25) .. controls (460,469.24) and (468.37,455.18) .. (480.96,447.2) -- (505,485.25) -- cycle ; \draw  [orient, color={rgb, 255:red, 74; green, 144; blue, 226 }  ,draw opacity=1 ] (517.93,528.36) .. controls (513.84,529.59) and (509.5,530.25) .. (505,530.25) .. controls (480.15,530.25) and (460,510.1) .. (460,485.25) .. controls (460,469.24) and (468.37,455.18) .. (480.96,447.2) ;  
\draw  [draw opacity=0] (480.3,447) .. controls (503.84,458.27) and (520.1,482.32) .. (520.1,510.16) .. controls (520.1,516.63) and (519.22,522.9) .. (517.57,528.85) -- (450.1,510.16) -- cycle ; \draw  [orient, color={rgb, 255:red, 144; green, 19; blue, 254 }  ,draw opacity=1 ] (480.3,447) .. controls (503.84,458.27) and (520.1,482.32) .. (520.1,510.16) .. controls (520.1,516.63) and (519.22,522.9) .. (517.57,528.85) ;  

\draw (426,478) node [anchor=north west][inner sep=0.75pt]    {$Q_{b}$};
\draw (518,449) node [anchor=north west][inner sep=0.75pt]    {$r( Q_{b})$};
\draw (359,558.25) node [anchor=north west][inner sep=0.75pt]    {$C^{+}$};
\draw (531,506) node [anchor=north west][inner sep=0.75pt]    {$C_{Q_{b}}{}$};
\draw (478.25,575.75) node [anchor=north west][inner sep=0.75pt]    {$e$};

    \end{tikzpicture}
\end{subfigure}
    \caption{A forward jump $Q_f$ and a backward jump $Q_b$ (in blue). The respective reaches $r(Q_f)$ and $r(Q_b)$ of the jumps are colored in purple and the cycles $C_{Q_f}$ and $C_{Q_b}$ are highlighted in yellow. 
    }
    \label{fig:jumps_reach}
\end{figure}

Finally, we introduce the notion of the \emph{reach} of a jump $Q$, which is intuitively the sub-path of $\Cp$ the jump $Q$ is ``jumping over''. \cref{fig:jumps_reach} illustrates the reach of a forward jump and a backward jump. We say that the jump $Q$ \emph{covers} the edges in the reach of $Q$.
\begin{definition*}[Reach of a jump]
Let $e \in \Ep(\Cp)$
and let $Q$ be a jump of $C_e$.
    \begin{itemize}
        \item If $Q$ is a {forward} jump: the reach of $Q$ is the sub-path of $\Cp$ defined as $r(Q) := \Cp \setminus C_{Q}$.
        \item If $Q$ is a {backward} jump: the reach of $Q$ is the sub-path of $\Cp$ defined as $r(Q) := C_Q \setminus Q$.
    \end{itemize}
\end{definition*}

 We also extend the definition of the reach of a jump to the reach of $C_e$. Intuitively the reach of $C_e$ is the union of sub-paths of $\Cp$ the cycle $C_e$ is jumping over. Here we differentiate between the sub-paths that are jumped over using \emph{forward motions} (forward jumps or interjumps) and \emph{backward motions} (backward jumps). 
\begin{definition*}[Reach of $C_e$]
    Let $e \in \Ep(\Cp)$.
    The forward reach $r_f(C_e)$ of $C_e$ is  defined as the union of every interjump  of $C_e$ and the reach of every forward jump of $C_e$.
    $$r_f(C_e) :=  \bigcup_{Q \in \mathcal{J}_e^f} r(Q) \cup \bigcup_{P \in \mathcal{I}_e} P $$
    The backward reach $r_b(C_e)$ of $C_e$ is defined as the union of the reaches of every backward jump of $C_e$.
    $$r_b(C_e) :=  \bigcup_{Q \in \mathcal{J}_e^b} r(Q)$$
    The reach $r(C_e)$ of $C_e$ is defined as the union of the forward reach and the backward reach of $C_e$. 
    $$r(C_e) := r_f(C_e)  \cup r_b(C_e)$$
\end{definition*}

All above definitions still hold when replacing $C_e$ by a sub-path $\mathcal{P}$ of $C_e$. Since the reach of a truncated jump is undefined, we will only consider sub-paths which endpoints are vertices of $\Cp$. If $\mathcal{P}$ contains a sub-path of an interjump then the forward reach of $\mathcal{P}$ contains this sub-path and not the entire interjump. 

With those newly defined notions, we note that to achieve \cref{cond:4single_path}, it suffices to prove the following \cref{lem:exists_Qb}.
\begin{lemma}\label{lem:exists_Qb}
    Let $\tup$ be a critical tuple. 
    Every edge $e \in \Ep(\Cp)$ is contained in the backward reach of $C_e$, i.e.\ $\forall e \in \Ep(\Cp):\, e \in r_b(C_e)$. 
\end{lemma}
Indeed, if $e \in r_b(C_e)$, then there exists a backward jump $Q_b$ in $C_e$ such that $e \in r(Q_b)$.
By the above definitions, the cycle $C_{Q_b}$ is a non-positive cycle containing $e$ and intersecting $\Cp$ in a single path. Hence the set of all these backward jumps is a set of cycles satisfying both \cref{cond:1cover_ep,cond:4single_path}. 
The remainder of this subsection is devoted to the proof of \cref{lem:exists_Qb}. The proof technique is a combination of reasonings about forward and backward motion, shown in \cref{lem:lemma8,lem:non_inter_Q}, and an interpolation argument.

Let's start with a basic fact about sub-paths of $C_e$. Let $e_s, e_f$ be two distinct edges on $\Cp$. The following lemma captures the intuitive fact that any path from $e_s$ to $e_f$ must go either in a clockwise motion, or in an anticlockwise motion. While doing so, it must traverse all the edges between $e_s$ and $e_f$ in that direction. 

\begin{lemma}\label{lem:lemma8}
    Let $\tup$ be a critical tuple and fix an edge $e \in \Ep(\Cp)$. If $e_s, e_f$ are distinct edges in $\Cp$ and $\mathcal{P} = C_e[e_s, e_f]$ is a sub-path of $C_e$ from $e_s$ to $e_f$, then either $\Cp[e_s, e_f] \subset r_f(\mathcal{P})$ or $\Cp[e_f, e_s] \subset r_b(\mathcal{P})$.
\end{lemma}

\begin{proof}
As a helpful tool, we consider the following definition of \emph{shadow}. Let $Q$ be a jump of $C_e$ from $x$ to $y$, where $x$ and $y$ are vertices in $\Cp$. 
If $Q$ is a forward jump, then the shadow $s(Q)$ of $Q$ is the walk in $\Cp$, which starts at $x$ and goes anticlockwise, until it encounters $y$. 
If $Q$ is a backward jump, then $s(Q)$ is the walk in $\Cp$, which starts at $x$ and goes clockwise, until it encounters $y$. We remark that formally, since $s(Q)$ can travel edges in reverse direction, $s(Q)$ is a walk in the undirected analogue of $\Cp$ in $G$.
For an interjump $P$ from vertex $x$ to $y$, the shadow of $P$ is defined to be the path $P$ itself. 
For a sub-path $\mc P \subset C_e$ consisting of the concatenation of alternating jumps and interjumps the shadow $s(\mc P)$ of $\mc P$ is the concatenation of the corresponding shadows of jumps and interjumps in the same order.

By definition, the following basic observations hold. For any sub-path $\mathcal{P} \subseteq C_e$, its shadow $s(\mathcal{P})$ is a walk contained in $\Cp$ such that it has the same start and end vertex as $\mathcal{P}$. Furthermore, the walk $s(\mathcal{P})$ can be thought of as a sequence of steps, where every step traverses one single edge of $\Cp$ either clockwise or anticlockwise. 
The set $r_f(\mathcal{P})$ is exactly the set of edges which are traveled by the walk $s(\mathcal{P})$ with an anticlockwise step. The set $r_b(\mathcal{P})$ is exactly the set of edges which are traveled by the walk $s(\mathcal{P})$ with a clockwise step. 
The set $r(\mathcal{P})$ is exactly the set of edges which are traveled by the walk $s(\mathcal{P})$.
    
Now return to the statement of the lemma. The walk $s(\mathcal{P})$ is a walk contained in $\Cp$, starting with the edge $e_s$ and ending with the edge $e_f$. It is obvious that $s(\mathcal{P})$ must either traverse all of $\Cp[e_s, e_f]$ with anticlockwise steps, or all of $\Cp[e_f, e_s]$ with clockwise steps. This implies that either $\Cp[e_s, e_f] \subset r_f(\mathcal{P})$ or $\Cp[e_f, e_s] \subset r_b(\mathcal{P})$.  
\end{proof}

In the proof of \cref{lem:exists_Qb}, we will consider sets of forward jumps in $C_e$ that cover all the cycle $\Cp$. However we want to avoid forward jumps covering the edge $e$ and thus consider them as ``bad'' jumps. The following \cref{lem:non_inter_Q} shows that under certain assumptions we can still find a set of forward jumps covering all of $\Cp$ without any ``bad'' jump in this set. 

\begin{lemma}\label{lem:non_inter_Q}
    Let $\tup$ be a critical tuple and fix an edge $e \in \Ep(\Cp)$. If $e \notin r_b(C_e)$, then there exists a set of forward jumps $\mc Q \subset \mathcal{J}_e^f$ such that 
    $$\Cp = \bigcup_{Q \in \mc Q} r(Q) \cup \bigcup_{P \in \mathcal{I}_e} P$$
    and $\nexists Q_f \in \mc Q$ with $e \in r(Q_f)$. 
\end{lemma}
\begin{proof}
    We prove the statement by constructing a tuple $(\mc P, \mc R, e')$, where $e' \in \Cp$ is an edge of $\Cp$, $\mc P = C_e[e, e']$ is a sub-path of $C_e$ and $\mc R$ is a set of jumps and interjumps of $C_e$. We want $(\mc P, \mc R, e')$ to satisfy the following invariant.
    \begin{invariant*}
    A tuple $(\mc P, \mc R, e')$ satisfies the invariant if the following properties hold:
    \begin{enumerate} 
        \item There is no forward jump $Q$ in $\mc R$ such that $e \in r(Q)$.\label[inv]{itm:inv1}
        \item $\Cp[e, e'] \subset r_f(\mc P)$.\label[inv]{itm:inv2}
        \item $\Cp(e', e) \subset r_f(\mc R)$ \label[inv]{itm:inv3}
    \end{enumerate}
    \end{invariant*}
    We also want to ensure that there is no forward jump $Q$ in $\mc P$ with $e\in r(Q)$. 
    Indeed, in that case, define $\mc Q$ to be the set of forward jumps in $\mc P \cup \mc R$. Then by \cref{itm:inv1}, $\mc Q$ doesn't contain any forward jump $Q$ with $e \in r(Q)$. Also, by \cref{itm:inv3,itm:inv2} of $(\mc P, \mc R, e')$:
    \begin{align*}
        \Cp &= \Cp[e, e'] \cup \Cp(e', e)\\
            &\subset r_f(\mc P) \cup r_f(\mc R)\\
            &\subset \bigcup_{Q \in \mc Q} r(Q) \cup \bigcup_{P \in \mathcal{I}_e} P
    \end{align*}
    which is the conclusion of \cref{lem:non_inter_Q}. 

    We start with a tuple $(\mc P, \mc R, e')$ satisfying the invariant properties and then iteratively modify the tuple in order to achieve the additional property that no forward jump in $\mc P$ covers $e$.
    At first, $e'$ is the edge before $e$ in $C_e$. Note that $e' \in \Cp$ because $e \notin M$ and $C_e$ is $M$-alternating so $e' \in M$, and since $\Cp$ is also $M$-alternating, edges of $M$ adjacent to $e$ must be in $\Cp$. Let $\mc P = C_e[e, e']$ be a path starting in $e$ and covering all edges of $C_e$ and let $\mc R = \emptyset$ be an empty set. 
    Clearly \cref{itm:inv1,itm:inv3} holds for $(\mc P, \mc R, e')$. By \cref{lem:lemma8} applied on $\mc P$, either $\Cp[e, e'] \subset r_f(\mc P)$ or $\Cp[e', e] \subset r_b(\mc P)$. However $e \notin r_b(C_e)$ so the latter is not possible. Hence $(\mc P, \mc R, e')$ satisfies all invariant properties. 
    Let $t$ be the number of forward jumps $Q$ in $\mc P$ with $e \in r(Q)$. If $t=0$, we are done.

    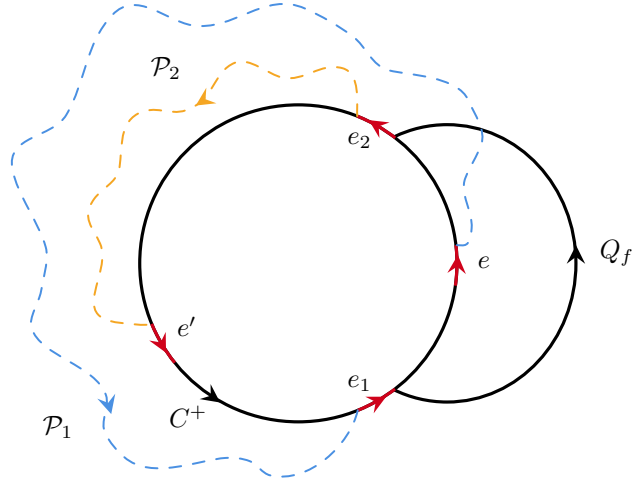
\begin{figure}[h]
    \centering
    
    \begin{tikzpicture}[scale=1]
    
\draw  [orient={<}{}{0.85}] (90,2560) .. controls (90,2515.82) and (125.82,2480) .. (170,2480) .. controls (214.18,2480) and (250,2515.82) .. (250,2560) .. controls (250,2604.18) and (214.18,2640) .. (170,2640) .. controls (125.82,2640) and (90,2604.18) .. (90,2560) -- cycle ;
\draw  [draw opacity=0] (249.48,2550.8) .. controls (249.82,2553.82) and (250,2556.89) .. (250,2560) .. controls (250,2563.83) and (249.73,2567.59) .. (249.21,2571.27) -- (170,2560) -- cycle ; \draw  [orient={<}{}{0.7}, color={rgb, 255:red, 208; green, 2; blue, 27 }  ,draw opacity=1 ] (249.48,2550.8) .. controls (249.82,2553.82) and (250,2556.89) .. (250,2560) .. controls (250,2563.83) and (249.73,2567.59) .. (249.21,2571.27) ;  
\draw  [draw opacity=0] (218.46,2496.08) .. controls (226.56,2492.17) and (235.55,2490) .. (245,2490) .. controls (280.9,2490) and (310,2521.34) .. (310,2560) .. controls (310,2598.66) and (280.9,2630) .. (245,2630) .. controls (235.38,2630) and (226.25,2627.75) .. (218.03,2623.71) -- (245,2560) -- cycle ; \draw [orient]  (218.46,2496.08) .. controls (226.56,2492.17) and (235.55,2490) .. (245,2490) .. controls (280.9,2490) and (310,2521.34) .. (310,2560) .. controls (310,2598.66) and (280.9,2630) .. (245,2630) .. controls (235.38,2630) and (226.25,2627.75) .. (218.03,2623.71) ;  
\draw  [draw opacity=0] (199.65,2485.68) .. controls (206.59,2488.45) and (213.06,2492.16) .. (218.88,2496.66) -- (170,2560) -- cycle ; \draw  [orient={<}{}{0.7}, color={rgb, 255:red, 208; green, 2; blue, 27 }  ,draw opacity=1 ] (199.65,2485.68) .. controls (206.59,2488.45) and (213.06,2492.16) .. (218.88,2496.66) ;  
\draw  [draw opacity=0] (218.73,2623.45) .. controls (213.03,2627.84) and (206.72,2631.47) .. (199.95,2634.2) -- (170,2560) -- cycle ; \draw  [orient={<}{}{0.7}, color={rgb, 255:red, 208; green, 2; blue, 27 }  ,draw opacity=1 ] (218.73,2623.45) .. controls (213.03,2627.84) and (206.72,2631.47) .. (199.95,2634.2) ;  
\draw  [draw opacity=0] (107.69,2610.17) .. controls (102.92,2604.27) and (98.99,2597.67) .. (96.04,2590.56) -- (170,2560) -- cycle ; \draw  [orient={<}{}{0.7}, color={rgb, 255:red, 208; green, 2; blue, 27 }  ,draw opacity=1 ] (107.69,2610.17) .. controls (102.92,2604.27) and (98.99,2597.67) .. (96.04,2590.56) ;  
\draw [orient={>}{}{0.4}, color={rgb, 255:red, 245; green, 166; blue, 35 }  ,draw opacity=1 ][line width=0.75] [line join = round][line cap = round] [dash pattern={on 4.5pt off 4.5pt}]  (199.65,2485.68) .. controls (199.44,2479.55) and (201.56,2464.04) .. (192.25,2461.5) .. controls (180.97,2458.42) and (173.03,2468.5) .. (162,2466.25) .. controls (154.68,2464.76) and (145.61,2455.71) .. (138,2458.75) .. controls (126.72,2463.26) and (128.03,2477.72) .. (113.25,2482.75) .. controls (104.74,2485.64) and (97.99,2479.65) .. (89,2484.5) .. controls (78.79,2490.01) and (81.51,2504.18) .. (79.5,2513) .. controls (78.08,2519.22) and (70.7,2522.3) .. (67,2527.5) .. controls (60.15,2537.12) and (66.6,2542.15) .. (70,2550) .. controls (76.41,2564.8) and (53.4,2587.31) .. (80,2590) .. controls (83.54,2590.36) and (92.71,2591.81) .. (96.04,2590.56) ;
\draw [orient={<}{}{0.25}, color={rgb, 255:red, 74; green, 144; blue, 226 }  ,draw opacity=1 ][line width=0.75] [line join = round][line cap = round] [dash pattern={on 4.5pt off 4.5pt}]  (199.95,2634.2) .. controls (196.15,2651.32) and (175.46,2671.53) .. (155.8,2667.01) .. controls (151.69,2666.07) and (149.32,2661.36) .. (145.4,2659.81) .. controls (141.64,2658.34) and (137.43,2658) .. (133.4,2658.21) .. controls (120.43,2658.89) and (96.01,2669.9) .. (82.2,2657.81) .. controls (78.34,2654.44) and (74.85,2650.41) .. (72.6,2645.81) .. controls (69.75,2640) and (74.93,2632.24) .. (73,2626.61) .. controls (67.7,2611.16) and (40,2607.12) .. (35.4,2592.61) .. controls (33.78,2587.52) and (39.31,2579.49) .. (40.2,2574.21) .. controls (41.29,2567.77) and (44.36,2558.11) .. (42.6,2551.81) .. controls (39.03,2539.02) and (28.29,2525.66) .. (25.8,2512.61) .. controls (22.46,2495.14) and (45.96,2497.01) .. (53,2485.01) .. controls (62.98,2468.01) and (62.34,2449.34) .. (76.2,2434.61) .. controls (88.33,2421.72) and (103.11,2437.56) .. (116.2,2439.81) .. controls (137.16,2443.43) and (152.86,2422.18) .. (173.4,2431.41) .. controls (181.85,2435.21) and (189.97,2439.88) .. (197.4,2445.41) .. controls (204.5,2450.7) and (206.8,2460.3) .. (215.8,2462.61) .. controls (225.39,2465.08) and (250.39,2465.28) .. (255.8,2477.81) .. controls (258.88,2484.94) and (264.85,2492.92) .. (262.2,2500.21) .. controls (260.08,2506.04) and (256.45,2511.23) .. (254.2,2517.01) .. controls (251.13,2524.9) and (258.06,2533.77) .. (257.4,2542.21) .. controls (256.58,2552.73) and (253.7,2550.8) .. (249.48,2550.8) ;

\draw (39.4,2635.4) node [anchor=north west][inner sep=0.75pt]    {$\mc P_{1}$};
\draw (103,2630) node [anchor=north west][inner sep=0.75pt]    {$C^{+}$};
\draw (258.8,2555) node [anchor=north west][inner sep=0.75pt]    {$e$};
\draw (193,2492.2) node [anchor=north west][inner sep=0.75pt]    {$e_{2}$};
\draw (193,2615) node [anchor=north west][inner sep=0.75pt]    {$e_{1}$};
\draw (320.2,2546.6) node [anchor=north west][inner sep=0.75pt]    {$Q_{f}$};
\draw (107,2585) node [anchor=north west][inner sep=0.75pt]    {$e'$};
\draw (93.8,2454.2) node [anchor=north west][inner sep=0.75pt]    {$\mc P_{2}$};

    \end{tikzpicture}

\caption{
To prove \cref{lem:non_inter_Q} we consider a path $\mathcal{P} = C_e[e, e']$ that contains a jump $Q_f$ with $e \in r(Q_f)$. $e_1$ and $e_2$ are the edges before and after the jump $Q_f$. We split $\mathcal{P}$ into two paths $\mathcal{P}_1$ (blue) and $\mathcal{P}_2$ (orange) that are respectively before and after the jump $Q_f$. 
}\label{fig:l6_usel8}
\end{figure}

    If $t > 0$, we replace $(\mc P, \mc R, e')$ by another tuple  $(\mc P', \mc R', e_1)$ where $\mc P'$ contains $t-1$ forward jumps covering $e$. We ensure that $(\mc P', \mc R', e_1)$ also satisfies the invariant properties, so that we can repeat the process until $t =0$. 
    Let $Q_f$ be the last forward jump in $\mc P$ covering $e$ and let $e_1$ and $e_2$ be the first edges in $\Cp$ before and after $Q_f$ (see \cref{fig:l6_usel8}). Consider the sub-paths $\mc P_1 = C_e[e, e_1]$ and $\mc P_2 = C_e[e_2, e']$ and define $\mc P' := \mc P_1$ and $\mc R' := \mc R \cup \mc P_2$. 
    By definition, $\mc P'$ is the sub-path of $\mc P$ before $Q_f$ (excluded), hence it contains $t-1$ forward jumps covering $e$. 
    It remains to show that $(\mc P', \mc R', e_1)$ satisfies the invariant. 
    By definition of $Q_f$, there is no forward jump in $\mc P_2$ covering $e$. Thus, since
    $(\mc P, \mc R, e')$ satisfies \cref{itm:inv1}, $(\mc P', \mc R', e_1)$ also satisfies \cref{itm:inv1}.
    We can apply \cref{lem:lemma8} on $\mc P_1$ and get that either $\Cp[e, e_1] \subset r_f(\mc P_1)$ or $\Cp[e_1, e] \subset r_b(\mc P_1)$. However $e \notin r_b(C_e)$ so the latter case is not possible, hence $(\mc P', \mc R', e_1)$ satisfies \cref{itm:inv2}.
    Finally, to prove that $\Cp(e_1, e) \subset r_f(\mc R')$, consider two cases.
    If $\Cp(e_1, e) \subset \Cp(e', e)$ (as in \cref{fig:l6_usel8}) then $\Cp(e_1, e) \subset r_f(\mc R) \subset r_f(\mc R')$ by \cref{itm:inv3} of $(\mc P, \mc R, e')$.
    Otherwise if $\Cp(e', e) \subset \Cp(e_1, e)$, then the edges $e_2, e_1, e', e$ appear in that order on $\Cp$. We can thus write $\Cp(e_1, e) = \Cp(e_1, e'] \cup \Cp(e', e)$. By \cref{itm:inv3} of $(\mc P, \mc R, e')$, we know that $\Cp(e', e) \subset r_f(\mc R)$. Note that $\Cp(e_1, e'] \subset \Cp[e_2, e']$. Applying \cref{lem:lemma8} on $\mc P_2$ we get that either $\Cp[e_2, e'] \subset r_f(\mc P_2)$ or $\Cp[e', e_2] \subset r_b(\mc P_2)$. But $e \in \Cp[e', e_2]$ and $e \notin r_b(C_e)$ so the latter case is not possible. Hence $\Cp(e_1, e'] \subset \mc P_2$ and we thus have $\Cp(e_1, e) \subset r_f(\mc P_2) \cup r_f(\mc R) = r_f(\mc R')$. 
    So in both cases \cref{itm:inv3} holds for $(\mc P', \mc R', e_1)$. 
    Therefore, we can safely replace $(\mc P, \mc R, e')$ by $(\mc P', \mc R', e_1)$ and maintain the invariant.
\end{proof}

We now have collected all the ingredients to prove the main result of this subsection. While the previous lemmas were proven using an intuition on forward and backward motion, \cref{lem:exists_Qb} is proven using an interpolation argument.

\begin{proof}[Proof of \cref{lem:exists_Qb}]
For the sake of contradiction, assume that there exists an edge $e \in \Ep(\Cp)$ such that $e \notin r_b(C_e)$. 
Apply \cref{lem:non_inter_Q} and let $\mc Q'$ be the resulting set of forward jumps. By identifying every jump in $\mc Q'$ with its reach, we can apply \cref{lem:keep_atmost2} and get the resulting set $\mc Q \subset \mc Q'$. Consequently, $\Cp = \bigcup_{Q \in \mc Q} r(Q) \cup \bigcup_{P \in \mathcal{I}_e} P$, there is no jump $Q \in \mc Q$ covering $e$ and every edge $e' \in \Cp$ is contained in the reach of at most two jumps of $\mc Q$. 
Order the jumps in $\mc Q$ by the distance between $e$ and their reach and write $\mc Q = \{Q_1, Q_2, \dots, Q_t\}$. Define $\mc Q_{odd} = \{Q_1, Q_3, \dots\}$ and $\mc Q_{even} = \{Q_2, Q_4, \dots\}$ to be the set of jumps of respectively odd and even index. Since there is no forward jump in $\mc Q$ covering $e$, the first and the last jump in $\mc Q$ are non-intersecting (otherwise the reach of one of them would contain $e$). Additionally, the reach of a jump $Q_i$ can only intersect the reach of $Q_{i-1}$ and $Q_{i+1}$ (by the property obtained by \cref{lem:keep_atmost2}).
Hence the reach of any two jumps in $\mc Q_{odd}$ as well as the reach of any two jumps in $\mc Q_{even}$ are not intersecting. 

We can thus define two cycles $\Cp_{odd} = \Cp \cup \mc Q_{odd} \setminus \{r(Q_1), r(Q_3), \dots\} $ and $\Cp_{even} = \Cp \cup \mc Q_{even} \setminus \{r(Q_2), r(Q_4), \dots\} $. We show that at least one of them is non-positive using \cref{obs:2types_cycles}. Indeed, the number of positive edges in $\Cp_{odd}$ is
\begin{align*}
    |\Ep(\Cp_{odd})| &= |\Ep(\Cp)| + \sum_{Q \in \mc Q_{odd}} |\Ep(Q)| - \sum_{Q \in \mc Q_{odd}} |\Ep(r(Q))| \\
    &= |\Ep(\Cp)| - \sum_{Q \in \mc Q_{odd}} m_Q
\end{align*}
where we define $m_Q := |\Ep(r(Q))| - |\Ep (Q)|$. Assume for now that $\sum\limits_{Q\in \mc Q_{odd}} m_Q \geq \sum\limits_{Q\in \mc Q_{even}} m_Q$ so that $\sum\limits_{Q \in \mc Q} m_Q = \sum\limits_{Q\in \mc Q_{odd}} m_Q + \sum\limits_{Q\in \mc Q_{even}} m_Q \leq 2 \sum\limits_{Q\in \mc Q_{odd}} m_Q$. So we get $$|\Ep(\Cp_{odd})| \leq  |\Ep(\Cp)| - \frac{1}{2} \sum\limits_{Q \in \mc Q}m_Q$$
and we also have
\begin{align*}
    \sum_{Q \in \mc Q} m_Q &= \sum_{Q\in \mc Q} |\Ep(r(Q))|  - \sum_{Q \in \mc Q} |\Ep(Q)| \\
    &= \left( \sum_{Q\in \mc Q} |\Ep(r(Q))| + \sum_{P \in \mathcal{I}_e}|\Ep(P)| \right) - \left( \sum_{P \in \mathcal{I}_e}|\Ep(P)| + \sum_{Q \in \mc Q} |\Ep(Q)| \right) \\
    &\geq |\Ep(\Cp)| - \left( \sum_{P \in \mathcal{I}_e}|\Ep(P)| + \sum_{Q \in \mc Q} |\Ep(Q)| \right)  \tag*{\text{(by the construction of $\mc Q$)}} \\
    &\geq |\Ep(\Cp)| - |\Ep(C_e)| \\
\end{align*}
where the last inequality comes from the fact that $C_e$ can contain jumps outside of $\mc Q$. We finally get the following bound on the number of positive edges in $\Cp_{odd}$.
\begin{align*}
    |\Ep(\Cp_{odd})| &\leq |\Ep(\Cp)| - \frac{1}{2}(|\Ep(\Cp)| - |\Ep(C_e)|) \\
    &= \frac{1}{2}|\Ep(\Cp)| + \frac{1}{2}|\Ep(C_e)| \\
    &< \frac{1}{2}k + \frac{1}{2} \cdot \frac{1}{3}k \tag*{\text{(by \cref{claim:Cp_bounds,pty:negC})}}\\
    &=\frac{2}{3}k
\end{align*}
By \cref{pty:posC}, we can conclude that $\Cp_{odd}$ is a non-positive cycle. 
If $\sum_{Q\in \mc Q_{odd}} m_Q \leq \sum_{Q\in \mc Q_{even}} m_Q$ then the same argument on $\Cp_{even}$ shows that $\Cp_{even}$ is a non-positive cycle. 
In the following we assume that $\Cp_{odd}$ is a non-positive cycle, but the reasoning can be adapted for the case that $\Cp_{even}$ is non-positive by interchanging $odd$ by $even$. 

\input{figures/l6_sequence}
Define a sequence of cycles $C^{i + 1}_{odd} := C^{i}_{odd}  \cup Q_{2i + 1} \setminus r(Q_{2i + 1})$ for $i \in |\mc Q_{odd}|$ starting at $C^0_{odd} := \Cp$ and ending at $\Cp_{odd}$. See \cref{fig:l6_sequence} for an example. The constructed $C^{i}_{odd}$'s are simple cycles since the reaches in $\mc Q_{odd}$ are non-intersecting. The sequence starts with a positive cycle $\Cp$ and ends with a non-positive cycle $\Cp_{odd}$ so there must exists an index $i$ such that $C^{i}_{odd}$ is positive and $C^{i+1}_{odd}$ is non-positive. However the number of positive edges in the non-positive cycle $C^{i+1}_{odd}$ is 
\begin{align*}
    |\Ep(C^{i+1}_{odd})| &= |\Ep(C^{i}_{odd})| + |\Ep(Q_{2i+1})|  - |\Ep(r(Q_{2i+1})| \\
    &= |\Ep(C^{i}_{odd})| + |\Ep(C_{Q_{2i+1}})| - |\Ep(\Cp)| \\
    &> \frac{2}{3}k + \frac{2}{3}k - k \tag*{\text{(by \cref{pty:posC,claim:Cp_bounds})}} \\
    &=\frac{1}{3}k
\end{align*}
By \cref{pty:negC}, this implies that $C^{i+1}_{odd}$ is positive, but we previously argued that it is non-positive. Hence it is not possible that $e \notin r_b(C_e)$.
\end{proof}

With \cref{lem:exists_Qb}, we can now construct a set of non-positive cycles satisfying \cref{cond:1cover_ep,cond:4single_path}.

\begin{lemma}\label{lem:get14}
    Let $\tup$ be a critical tuple. There exists a set $\mathcal{C}$ of non-positive cycles in $G_M$ satisfying \cref{cond:1cover_ep,cond:4single_path}, i.e.\ such that the following holds:
    \begin{itemize}
        \item $\forall e \in \Ep(\Cp)$ there exists a cycle in $\mathcal{C}$ containing $e$.
        \item $\forall C \in \mathcal{C},\, C \cap \Cp$ is a single path.
    \end{itemize}
\end{lemma}
\begin{proof}
    By \cref{lem:exists_Qb}, for each edge $e \in \Ep(\Cp)$ there exists a backward jump in $C_e$, call it $Q_e$, that covered $e$. Since it is a backward jump, the cycle $C_{Q_e}$ formed by $Q_e$ and its reach is a non-positive cycle. It also contains $e$ and intersects $\Cp$ in a single path $r(Q_e)$. Hence the set $\mathcal{C} := \{C_{Q_e}: e \in \Ep(\Cp), Q_e \in \mathcal{J}_e^b, e \in r(Q_e)\}$ satisfies \cref{cond:1cover_ep,cond:4single_path}.
\end{proof}

\subsection{Modifying the set to satisfy Condition \ref{cond:34cover_en}}\label{subsec:reduce-cover}
In the previous section we obtained a set $\mathcal{C}$ of non-positive cycles satisfying \cref{cond:1cover_ep,cond:4single_path}. To additionally satisfy \cref{cond:34cover_en}, which is that for every negative edge $e \in \En(G_M)$ there exist at most two cycles in $\mathcal{C}$ containing $e$, we proceed in two steps. First we focus in \cref{lem:get_cond3} on bounding the number of cycles containing a negative edge on the cycle $\Cp$. Then we prove in \cref{lem:cycle_fusion} how to bound the number of cycles containing a negative edge outside the cycle $\Cp$. 

\begin{lemma}\label{lem:get_cond3}
Let $\tup$ be a critical tuple and $\mathcal{C}$ be a set of non-positive cycles in $G_M$ satisfying \cref{cond:1cover_ep,cond:4single_path}, i.e.\ such that:
    \begin{itemize}
        \item $\forall e \in \Ep(\Cp)$ there exists a cycle in $\mathcal{C}$ containing $e$.
        \item $\forall C \in \mathcal{C},\, C \cap \Cp$ is a single path.
    \end{itemize}
    Then there exists a set $\mathcal{C'} \subset \mathcal{C}$ of non-positive cycles satisfying the above \cref{cond:1cover_ep,cond:4single_path} as well as the following property:
    \begin{itemize}
        \item $\forall e \in \En(\Cp)$ there are at most two cycles in $\mathcal{C}$ containing $e$.
    \end{itemize}
\end{lemma}

\begin{proof}
    By \cref{cond:4single_path}, we can associate to each cycle $C$ of $\mathcal{C}$ the sub-path $P_C :=C \cap \Cp$ intersecting with $\Cp$. 
    Let $\mathcal{P} := \{P_C : C \in \mathcal{C}\}$ be the set of such intersecting paths. By \cref{cond:1cover_ep} for every edge $e \in \Ep(\Cp)$ there exists a path $P \in \mathcal{P}$ containing $e$. We can thus apply \cref{lem:keep_atmost2} to $\mathcal{P}$ and get a set $\mathcal{P}' \subset \mathcal{P}$ such that for every edge $e \in \Ep(\Cp)$ there exists a path $P \in \mathcal{P}'$ containing $e$  and every $e \in C^+$ is contained in at most two paths from $\mathcal{P}'$. 
    Let $\mathcal{C}' := \{C : P_C \in \mathcal{P}'\}$ be the set of cycles associated to each path in $\mathcal{P}'$. Then $\mathcal{C}' \subset \mathcal{C}$ still satisfies  \cref{cond:1cover_ep,cond:4single_path} and additionally for every edge $e \in C^+$ at most two cycles in $\mathcal{C}'$ contain $e$.
\end{proof}

\begin{lemma}\label{lem:cycle_fusion}
    Let $\tup$ be a critical tuple and $\mathcal{C}$ be a set of non-positive cycles in $G_M$ satisfying \cref{cond:1cover_ep,cond:4single_path,}, i.e.\ such that:
    \begin{itemize}
        \item $\forall e \in \Ep(\Cp)$ there exists a cycle in $\mathcal{C}$ containing $e$.
        \item $\forall C \in \mathcal{C},\, C \cap \Cp$ is a single path.
    \end{itemize}
    Then there exists a set $\mathcal{C'} \subset \mathcal{C}$ of non-positive cycles satisfying the above \cref{cond:1cover_ep,cond:4single_path} as well as the following property:
    \begin{itemize}
        \item $\forall e \in \En(G_M \setminus \Cp)$ there are at most two cycles in $\mathcal{C}$ containing $e$.
    \end{itemize}
\end{lemma}

\begin{proof}
    We show how to reduce three non-positive cycles of $\mathcal{C}$ containing an edge $e \in \En(G_M \setminus \Cp)$ into two non-positive cycles such that, when replacing the three cycles by the two new cycles, \cref{cond:4single_path,cond:1cover_ep} still hold. By repeating this transformation, the number of cycles containing a negative edge $e \in \En(G_M \setminus \Cp)$ decreases until eventually at most two cycles in $\mathcal{C}$ contain $e$. 
    
    Fix an edge $e \in \En(G_M \setminus \Cp)$ and three non-positive cycles $C_{Q_1}$, $C_{Q_2}$ and $C_{Q_3}$ in $\mathcal{C}$ all containing $e$. For $i \in \{1, 2, 3\}$, by \cref{cond:4single_path}, $C_{Q_i}$ intersects $\Cp$ in a single path. Hence let $r(Q_i) = C_{Q_i} \cap \Cp$ be this path and let $Q_i := C_{Q_i} \setminus \Cp$ be the remaining path of $C_{Q_i}$ not intersecting $\Cp$. Let also $s_i$ and $f_i$ be the first and last edge of $Q_{i}$.
    We can assume w.l.o.g.\ that the first edge of $r(Q_i)$ appears anti-clockwise on $\Cp$ in the order $1, 2, 3$. 
    
    For $(i, j) \in \{(1, 2), (2, 3), (3, 1)\}$ consider the walk $W_{ij} = Q_j[s_j, e] \cup Q_i(e, f_i]$ that goes from $s_j$ to $f_i$. This is a valid walk since $e \in  Q_1 \cap Q_2 \cap Q_3$. $W_{ij}$ can use edges multiple times so we identify $W_{ij}$ to its multi-edge set. As illustrated in \cref{fig:walk_to_path}, we can extract a simple path $P_{ij} \subset W_{ij}$ that goes from $s_j$ to $f_i$. Let $C_{ij} = P_{ij} \cup \Cp[r(Q_i), r(Q_j)]$. $C_{ij}$ is a simple directed cycle since $P_{ij}$ is a simple path from $s_j$ to $f_i$ outside of $\Cp$ and $\Cp[r(Q_i), r(Q_j)]$ is a simple path on $\Cp$ from the endpoint of $f_i$ on $\Cp$ to the endpoint of $s_j$ on $\Cp$. 
    An example of this construction is shown in \cref{fig:cycle_fusion}.  If $C_{ij}$ is a non-positive cycle, then we can replace $C_{Q_i}$ and $C_{Q_j}$ by $C_{ij}$ in $\mathcal{C}$ and \cref{cond:4single_path,cond:1cover_ep} would still hold. Indeed, by construction $C_{ij}$ intersects $\Cp$ on the single path $\Cp[r(Q_i), r(Q_j)]$ so \cref{cond:4single_path} holds. \cref{cond:1cover_ep} holds because every edge $e \in \Ep(\Cp)$ that is contained in $C_{Q_i}$ or $C_{Q_j}$ is in fact in $r(Q_i)$ or $r(Q_j)$, which are both contained in $C_{ij}$. Hence it remains to show that there exists a pair $(i, j) \in \{(1, 2), (2, 3), (3, 1)\}$ such that $C_{ij}$ is a non-positive cycle. 
    \input{figures/walk_to_path}
    \input{figures/cycle_fusion}
    Suppose there exists a pair $(i, j) \in \{(1, 2), (2, 3), (3, 1)\}$ such that $r(Q_i)$ and $r(Q_j)$ are intersecting. Then we have:
    \begin{align*}
        |\Ep(C_{ij})| &= |\Ep(P_{ij})| + |\Ep(\Cp[r(Q_i), r(Q_j)])| \\
        &= |\Ep(P_{ij})| + |\Ep(r(Q_i) \cup r(Q_j))|\\
        &\leq |\Ep(W_{ij})| + |\Ep(r(Q_i))| + |\Ep(r(Q_j))| \tag*{\text{(because $P_{ij} \subset W_{ij}$)}} \\
        &= |\Ep(Q_j[s_j, e])| + |\Ep(Q_i(e, f_i])| + |\Ep(r(Q_i))| + |\Ep(r(Q_j))|  \\
        &\leq |\Ep(Q_j)| + |\Ep(Q_i)| + |\Ep(r(Q_i))| + |\Ep(r(Q_j))|  \\
        &= |\Ep(C_{Q_i})| + |\Ep(C_{Q_j})|  \\
        &< 2 \cdot \frac{1}{3}k \tag*{(by \cref{pty:negC})} \\
    \end{align*}
    where the second equality follows from $\Cp[r(Q_i), r(Q_j)] \subset r(Q_i) \cup r(Q_j)$.
    By \cref{pty:posC}, $C_{ij}$ cannot be positive and hence is a non-positive cycle, so we get the desired result.

    If none of the paths $r(Q_1), r(Q_2), r(Q_3)$ are intersecting then assume for the sake of contradiction that $C_{12}, C_{23}$ and $C_{31}$ are all positive. By \cref{pty:posC} this means that 
    $$|\Ep(C_{12})| + |\Ep(C_{23})| + |\Ep(C_{31})| > 3 \cdot \frac{2}{3}k = 2k.$$
    However we also have for any pair $(i, j) \in \{(1, 2), (2, 3), (3, 1)\}$ that
    \begin{align*}
        |\Ep(C_{ij})| &= |\Ep(P_{ij})| + |\Ep(\Cp[r(Q_i), r(Q_j)])| \\
        &\leq |\Ep(W_{ij})| + |\Ep(\Cp[r(Q_i), r(Q_j)])| \tag*{\text{(because $P_{ij} \subset W_{ij}$)}}\\
        &= |\Ep(Q_j[s_j, e])| + |\Ep(Q_i(e, f_i])| + |\Ep(\Cp[r(Q_i), r(Q_j)])|
    \end{align*}
    Hence the total number of positive edges in the three cycles is:
    \begin{align*}
        & |\Ep(C_{12})| + |\Ep(C_{23})| +|\Ep(C_{31})| \\
        &\quad \leq |\Ep(Q_2[{s_2}, e])| + |\Ep(Q_1(e, {f_1}])| + |\Ep(\Cp[r(Q_1), r(Q_2)])| \\
        &\quad \quad + |\Ep(Q_3[{s_3}, e])| + |\Ep(Q_2(e, {f_2}])| + |\Ep(\Cp[r(Q_2), r(Q_3)])| \\
        &\quad \quad + |\Ep(Q_1[{s_1}, e])| + |\Ep(Q_3(e, {f_3}])| + |\Ep(\Cp[r(Q_3), r(Q_1)])| \\
        &\quad = |\Ep(Q_1)| + |\Ep(Q_2)| + |\Ep(Q_3)| \\
        &\quad \quad + |\Ep(\Cp)| + |\Ep(r(Q_1))| + |\Ep(r(Q_2))| + |\Ep(r(Q_3))| \\
        &\quad = |\Ep(\Cp)| + |\Ep(C_{Q_1})| + |\Ep(C_{Q_2})| + |\Ep(C_{Q_3})| \\ 
        &\quad < k + 3 \cdot \frac{1}{3}k = 2k \tag*{\text{(by  \cref{claim:Cp_bounds,pty:negC})}}
    \end{align*}
    where we can apply \cref{pty:negC} because $C_{Q_i}$ is non-positive for any $i \in \{1, 2, 3\}$.
    Hence there has to exists a pair $(i, j) \in \{(1, 2), (2, 3), (3, 1)\}$ such that $C_{ij}$ is non-positive.
\end{proof}



\subsection{Constructing a target set}\label{subsec:summarize-proof-at-the-end}
We can now prove \cref{lem:good_Me}.
\begin{proof}[Proof of \cref{lem:good_Me}]
    For the sake of contradiction let $\tup$ be a critical tuple. Apply Lemmas \ref{lem:get14}, \ref{lem:cycle_fusion} and \ref{lem:get_cond3} in that order consecutively and get a target set $\mathcal{C}$. By \cref{lem:contradiction}, this leads to a contradiction. Thus $\tup$ cannot be a critical tuple.
\end{proof}
As a consequence, \cref{lem:good_Me} and \cref{lem:correctness} together show that
\cref{algo:3approx} is a polynomial-time algorithm that  always outputs a perfect matching $M$ containing between $\frac{1}{3}k$ and $k$ red edges. This completes the proof of \cref{thm:3approx}.

\section{Conclusion}
In this paper we formally define {\EMO}, an optimization variant of the \textsc{Exact Matching} problem and show a deterministic polynomial-time approximation algorithm for {\EMO} which achieves an approximation ratio of 3 on bipartite graphs. Although the algorithm is fairly simple to present, the proof of its correctness is more complex. In the second part of the algorithm we iterate over all edges $e$ and compute a perfect matching of minimum number of red edges containing $e$. Most of our work is put into proving that there exists such a matching that approximates the optimal solution of {\EMO} by a factor 3. As our calculations are tight for an approximation ratio of 3, a natural continuation of our work would be to improve this ratio, e.g.\ to 2.
We speculate that the algorithm could be improved by iterating over larger subsets of edges (of constant size) instead of one edge at a time and that such an algorithm could give a better approximation and maybe even a PTAS. The analysis of such an algorithm, however, remains quite challenging and new techniques and insights might be needed. Another open problem is to find an approximation algorithm for general graphs.

\bibliography{ref.bib}
\end{document}